\documentclass[11pt]{article}

\usepackage[T1]{fontenc}
\usepackage[utf8]{inputenc}
\usepackage{lmodern}
\usepackage{amsmath,amssymb,amsthm,mathtools,bm}
\usepackage{geometry}
\usepackage{booktabs}
\usepackage{csquotes}
\usepackage{array}
\usepackage{tabularx}
\usepackage{enumitem}
\usepackage{tikz}
\usetikzlibrary{arrows.meta}
\usepackage[hidelinks]{hyperref}
\usepackage[nameinlink,noabbrev]{cleveref}
\usepackage[
  backend=biber,
  style=numeric-comp,
  sorting=none,
  giveninits=true,
  maxbibnames=5,
  maxcitenames=2,
  eprint=true,
  doi=true,
  url=false,
  isbn=false
]{biblatex}
\newtheorem{theorem}{Theorem}
\newtheorem{proposition}{Proposition}

\newtheorem{remark}{Remark}

\newcommand{\dd}{\mathrm{d}}

\newcommand{\WB}{\mathrm{WB}}

\newcommand{\diag}{\mathrm{diag}}

\DeclareMathOperator{\tr}{tr}
\DeclareMathOperator{\spanop}{span}

\title{Tamm--Rubilar branch diagnostics for Drummond--Hathrell photon propagation: Schwarzschild calibration and a Kerr weak-lensing benchmark}
\author{Jos\'e Rodal\\Rodal Consulting, Cary, NC 27513, USA\\\texttt{jrodal@alum.mit.edu}}
\date{\today}

\begin{document}
\maketitle

\begin{abstract}
This paper develops a local Tamm--Rubilar branch diagnostic for supplied
nondispersive, pair-symmetric constitutive tensors and uses it as a
preprocessing layer for polarization-resolved ray transport.  The diagnostic
constructs the local quartic \(P_x(q)={\cal G}^{abcd}q_aq_bq_cq_d\), certifies
real ADM-oriented branches and root margins, and separates algebraic branch
stability from the separate effective-field-theory question of when a
Drummond--Hathrell (DH) low-frequency surrogate is applicable.

For constitutive tensors with a reflection isometry, the adapted bivector matrix
has the exact block form \(G_{++}\oplus G_{--}\).  The resulting full
parity-invariant Tamm--Rubilar polynomial is a twelve-variable quartic with
coefficients even in the transverse momentum.  A restricted six-variable
meridional SSSW-frame quartic is retained as a compact analytic benchmark and
root-margin test, while generic supplied tensors are evaluated through the full
invariant contraction.

The physical calibration is the Ricci-flat DH curvature sector.  In
Schwarzschild, the full parity tensor factorizes and reproduces the standard
radial/no-shift and orbital/split low-frequency result.  The corresponding
first-order screen reduction gives branch Hamiltonians, a phase/group-delay
formula at fixed asymptotic Killing frequency, and a branch-labelled retardance
for a specified polarized input.

The flagship rotating-spacetime benchmark is an infinity-to-infinity weak-lensing
calculation in the linearized Kerr field.  Daniels and Shore established the
local Kerr DH propagation problem; here the complementary asymptotic question is
answered in a transported Born screen.  The local slow-Kerr magnetic-Weyl angle
\(\chi_K=-3a\cos\theta/(2r)+O(a^3/r^3)\) is a principal-frame eigenbasis tilt,
whereas the signed endpoint eigenaxis mismatch cancels at leading order.  The
surviving branch-delay benchmark is
\[
 \Delta t_{+-}^{\rm Kerr,Born}=
 \frac{4\alpha}{45\pi}\lambda_e^2\frac{M}{b^2}
 \left[1+2\frac{\bm J\cdot(\hat{\bm k}\times\hat{\bm b})}{Mb}\right]
 +O\!\left(\zeta_{\rm DH}\frac{M^2}{b^3},
 \zeta_{\rm DH}\frac{J^2}{Mb^4},(\zeta_{\rm DH}R)^2\right),
\]
with \(\delta_{+-}=\omega_\infty\Delta t_{+-}\).  The explicit Riemann
projection from the linearized Kerr metric to the Born screen is given so the
spin-odd term is reproducible.  Astrophysical estimates give
\(\delta_{+-}\sim10^{-24}\)--\(10^{-23}\) at few-keV energies for representative
compact-object grazing rays; the Kerr calculation is therefore presented as a
scale-setting benchmark for the diagnostic and transport framework, not as a
near-term detectability claim.
\end{abstract}

\vspace{2pc}


\vspace{1pc}
\noindent\textit{Notation.}  The QED length scale is the reduced electron Compton wavelength $\lambda_e=m_e^{-1}$ in units $\hbar=c=1$.  The flat-slice symbols $\rho,\mathcal A,\mathcal B,\mathcal D$ are used only in the appendix-level tensorial testbed; $C_3,C_2,C_1,C_0$ denote quartic coefficients throughout.

\section{Scope and main results}
\label{sec:scope}

The manuscript has two linked aims.  The first is methodological: given a local
nondispersive pair-symmetric constitutive tensor \(G^{abcd}(x)\), supplied by an
independent model or effective action, construct the Tamm--Rubilar \cite{Tamm1925,Rubilar2002,HehlObukhov} polynomial
\(P_x(q)={\cal G}^{abcd}q_aq_bq_cq_d\) and certify whether its local Fresnel roots
are real, ADM-time oriented, branch-separated, and perturbatively close to the
metric double cone.  This pointwise algebraic screen is designed for automated
polarization ray codes, where it prevents integration across complex-root branch
jumps, near-degenerate branch swaps, or sign catastrophes before Hamiltonian
transport is attempted.

The second aim is a controlled physical calibration.  The Drummond--Hathrell (DH) \cite{DrummondHathrell,Shore2003}
curvature coupling supplies a fixed low-frequency input tensor on Ricci-flat
backgrounds.  Schwarzschild provides an exactly factorized local test, and a
weak-lensing Kerr calculation provides a rotating-spacetime benchmark for the
same branch-labelling and retardance machinery.  The framework is local and
algebraic; global hyperbolicity, gravitational closure, backreaction, and full
frequency-dependent QED propagation are separate layers.

Daniels and Shore \cite{DanielsShore1996} established the local DH photon-propagation problem in Kerr,
including velocity shifts, the polarization sum rule, and horizon/ergosphere
properties.  The question addressed here is complementary: after fixing an
asymptotic Sachs screen and a branch labelling for a two-ended weak-lensing
experiment, the local magnetic-Weyl eigenbasis tilt has zero leading endpoint
mismatch, while the branch-delay split has a reproducible spin-odd Born term.

The full algebraic object used below is the parity-reduced polynomial
\begin{equation}
 P_{\rm par}(\omega,x,y,z)=\omega^4+C_3\omega^3+C_2\omega^2+C_1\omega+C_0,
 \label{eq:intro-Ppar}
\end{equation}
where a reflection isometry gives the bivector block form
\(G_{++}\oplus G_{--}\) and makes each coefficient even in the transverse
momentum \(z\).  A restricted meridional SSSW\cite{SSSW2017}-frame polynomial,
\begin{equation}
 P_\Psi(\omega,x,y,z)=\omega^4+C_2\omega^2+C_1\omega+C_0,
 \label{eq:intro-Ppsi}
\end{equation}
is printed because its coefficients are short enough to audit.  It is a
sectoral benchmark, not a replacement for the transverse block.

For DH input the tensor is retained only through
\(O(\alpha R/m_e^2)\).  The roots are therefore low-frequency roots of the
supplied nondispersive surrogate.  They are not the high-frequency QED
wavefront characteristics, which require the frequency-dependent polarization
operator.  The paper consequently separates the following levels of statement:
\begin{center}
\small
\begin{tabularx}{\textwidth}{@{}p{0.31\textwidth}p{0.39\textwidth}X@{}}
\toprule
Claim & Status & Location \\
\midrule
Parity block form and even-in-\(z\) coefficients & exact algebraic theorem & Sec.~\ref{sec:full-parity-polynomial} \\
Six-variable SSSW quartic & analytic sectoral benchmark; expanded form in supplement & Sec.~\ref{sec:meridional-sector}; App.~\ref{app:expanded-coefficients} \\
Schwarzschild DH factorization & exact for the supplied surrogate; physical at first DH order & Sec.~\ref{sec:dh} \\
Ricci-flat screen branches & first-order DH propagation reduction & Sec.~\ref{sec:ricci-flat-propagation} \\
Kerr Born result & weak-lensing branch-delay/phase-retardance benchmark & Sec.~\ref{subsec:slow-kerr} \\
\bottomrule
\end{tabularx}
\end{center}

The rotating-spacetime benchmark is the Kerr weak-lensing branch-delay split
\eqref{eq:kerr-born-retardance}.  The convention used throughout is
branch-labelled: \(\lambda_+\) is the upper screen eigenvalue continued from the
Schwarzschild tangential branch, \(\lambda_-\) is the lower branch, and
\(\Delta t_{+-}:=t_- - t_+\).  Thus \(\Delta t_{+-}\) is a signed correction to a
positive branch separation in the weak-lensing domain, while \(|\Delta t_{12}|\)
denotes an unlabeled positive magnitude.  The dimensionless optical retardance
for a monochromatic input is \(\delta_{+-}=\omega_\infty\Delta t_{+-}\).  The
slow-Kerr magnetic-Weyl angle is retained as a local screen-eigenbasis diagnostic,
not as the scattering observable.

The formal area-metric setting is the one used in constructive gravity and
premetric electrodynamics \cite{SchullerWohlfarth2006,PunziSchullerWohlfarth2007,SchullerWitte2014,SSSW2017,HehlObukhov,Hehl2008MinkowskiPremetric}.  Recent work on quadratic area-metric actions and area-metric backgrounds emphasizes that area metrics carry non-length degrees of freedom and that their reduction toward ordinary metric geometry is a separate algebraic and dynamical question \cite{BorissovaDittrichKrasnov2024,BorissovaHo2024}.  Birefringent propagation and global optical relations provide the neighboring phenomenological setting \cite{RaetzelRiveraSchuller2011,Alex2020,SchullerWerner2017DistanceDuality}.  The operational contribution below is the local algebraic layer: construct \(P_x(q)\), certify roots and margins, attach the EFT status marker when the input is DH, and then pass accepted branches to ray transport.  Recent gravitational-birefringence estimates and modern strong-field polarimetry motivate such careful polarization transport in lensing environments \cite{MurkTernoVadapalli2025,Steiner2024IXPECygX1}; the numerical Kerr scale is recorded to delimit, not advertise, observational relevance.

\section{Meridional support and reflection bookkeeping}
\label{sec:meridional-bookkeeping}

The restricted SSSW polynomial used in \cref{sec:principal} is motivated by a
fixed meridional bivector support
\[
  \mathcal E=\spanop\{e_0\wedge e_1,\ e_0\wedge e_2,\ e_1\wedge e_2\}
  =\spanop\{[01],[02],[12]\}.
\]
For the flat-slice shift-potential testbed summarized in
\cref{app:shift-background}, the ADM-normal connection generically rotates a
rank-one meridional ansatz out of itself unless
\(\mathcal A\mathcal D-\mathcal B^2=0\).  At the equator of that testbed this
quantity is \(-v^2(g')^2\), so a generic transition region forces the full
three-dimensional meridional support.  This argument is only kinematic; it is
not a gravitational-closure equation.

The same reflection symmetry gives the important algebraic split.  Under
\(P_\phi:\phi\mapsto-\phi\), the adapted tetrad has \(e_3\mapsto-e_3\) while
\(e_0,e_1,e_2\) are fixed.  Hence
\[
 V_+=\spanop\{[01],[02],[12]\},\qquad
 V_-=\spanop\{[03],[13],[23]\}.
\]
Any pair-symmetric constitutive tensor that inherits the reflection has
\(G_{+-}=0\), but the transverse block \(G_{--}\) is symmetry-allowed and must be
included in the full Tamm--Rubilar polynomial.  This is why the printed
six-variable polynomial is used only as a compact meridional benchmark.

\section{Restricted meridional SSSW-frame diagnostic and local parity-invariant extension}
\label{sec:principal}
\label{sec:meridional-sector}

The meridional block maps to the sparse six-variable sector
\begin{equation}
 \Psi=(\phi_1,\phi_2,\phi_4,\phi_{12},\phi_{16},\phi_{17})
\end{equation}
by the normalized component coordinates
\begin{equation}
 h_{11}=\phi_1,
 \quad h_{12}=\frac{\phi_2}{\sqrt2},
 \quad h_{22}=\phi_4,
 \quad h_{33}=\phi_{12},
 \quad h_{23}=\frac{\phi_{16}}{\sqrt2},
 \quad h_{13}=\frac{\phi_{17}}{\sqrt2} .
 \label{eq:normalized-meridional-coordinates}
\end{equation}
Thus the off-diagonal symmetric modes use the usual $1/\sqrt2$ normalization.  Define
\begin{equation}
 u=1+\phi_1,
 \qquad s=\frac{\phi_2}{\sqrt2},
 \qquad w=1+\phi_4,
 \qquad \ell=1+\phi_{12},
 \qquad \Delta^2=uw-s^2 .
 \label{eq:sssw-definitions}
\end{equation}
Whenever \(\Delta\) appears without the square, it denotes the positive square root of \(uw-s^2\) on the symmetric-positive branch.
The upper and lower symmetric SSSW blocks are
\begin{equation}
 U=\begin{pmatrix}u&s&0\\s&w&0\\0&0&1\end{pmatrix},
 \qquad
 L=\begin{pmatrix}1&0&0\\0&1&0\\0&0&\ell\end{pmatrix} .
\end{equation}
For the mixed block write
\begin{equation}
 m_{17}=\frac{\phi_{17}}{\sqrt2},
 \qquad m_{16}=\frac{\phi_{16}}{\sqrt2},
 \qquad N=(u+1)(w+1)-s^2,
\end{equation}
and
\begin{equation}
 c=\frac{2[(w+1)m_{17}-s m_{16}]}{N},
 \qquad
 d=\frac{2[(u+1)m_{16}-s m_{17}]}{N},
 \qquad
 a=uc+sd,
 \qquad
 b=sc+wd .
\end{equation}
Then
\begin{equation}
 W=g^i{}_j=\begin{pmatrix}0&0&a\\0&0&b\\c&d&0\end{pmatrix} .
 \label{eq:mixed-block-W}
\end{equation}
satisfies the SSSW frame condition $WU=(WU)^T$.  The reconstruction is used only
on the symmetric-positive branch: $U>0$ and $\ell>0$.  On this branch
\(\Delta^2=uw-s^2>0\) and
\(N=(u+1)(w+1)-s^2=\Delta^2+u+w+1>0\), so the mixed-block denominator is not a
separate singularity.

\begin{theorem}[Restricted meridional SSSW-frame diagnostic]
\label{thm:restricted-sssw-polynomial}
\label{thm:restricted-polynomial}
Let the restricted meridional state be
\[
\Psi=(\phi_1,\phi_2,\phi_4,\phi_{12},\phi_{16},\phi_{17})
\]
and let \(u,s,w,\ell,\Delta,a,b,c,d\) be defined by
\eqref{eq:sssw-definitions}--\eqref{eq:mixed-block-W}.  Assume the admissible
symmetric-positive branch $U>0$ and $\ell>0$; hence \(\Delta^2>0\) and \(N>0\).
For a covector \(q_a=(\omega,x,y,z)\) in the orthonormal frame, direct
contraction of the SSSW principal tensor with the reconstructed six-variable
area metric gives a raw quartic whose leading coefficient is \(\Delta^2\).  On
this branch, multiplication by a positive nonzero scalar does not change the
characteristic set.  After monic normalization by this coefficient, the
restricted root polynomial is
\begin{equation}
 \boxed{P_\Psi(\omega,x,y,z)=\omega^4+C_2\omega^2+C_1\omega+C_0 .}
\end{equation}
The cubic term vanishes identically because the reconstruction imposes
\(WU=(WU)^T\).  In the metric specialization
\[
 u=w=\ell=1,\qquad s=a=b=c=d=0,
\]
one obtains
\[
P_\Psi=(\omega^2-x^2-y^2-z^2)^2 .
\]
\end{theorem}

\subsection{Parity-invariant extension to the full local constitutive sector}
\label{sec:full-parity-polynomial}

The restricted six-variable polynomial of
\cref{thm:restricted-sssw-polynomial} is a meridional sectoral diagnostic
associated with the parity-even variables.  It should be read together with the
full parity-invariant construction below: a complete local parity-reduced
characteristic analysis also requires the parity-odd block \(G_{--}\).

\begin{theorem}[Full local parity-invariant Tamm--Rubilar polynomial]
\label{thm:complete-parity-polynomial}
\label{thm:full-parity-tr}
Let the reflection \(P_\phi:\phi\mapsto-\phi\) act on the adapted
orthonormal tetrad by
\[
e_3\mapsto -e_3,\qquad e_0,e_1,e_2\mapsto e_0,e_1,e_2 .
\]
Then the bivector space decomposes as
\[
V_+=\operatorname{span}\{[01],[02],[12]\},\qquad
V_-=\operatorname{span}\{[03],[13],[23]\}.
\]
If the local constitutive area metric inherits this reflection symmetry, then
in the parity basis
\[
{\cal B}_{\rm par}=([01],[02],[12]\mid[03],[13],[23])
\]
(used locally away from the polar-axis tetrad singularity, with extension by
continuity on a regular axis) its six-by-six bivector matrix has the block form
\[
{\mathsf G}_{\rm par}=
\begin{pmatrix}G_{++}&0\\0&G_{--}\end{pmatrix},
\]
where \(G_{++}\) and \(G_{--}\) are independent symmetric \(3\times3\)
blocks.  Here \({\mathsf G}_{IJ}\) denotes the raw bivector component
\(G^{ab\,cd}\), with antisymmetry in each index pair and pair-exchange
symmetry.  For signature \((-+++)\) the metric vacuum is
\[
{\mathsf G}_{\rm vac}=
\diag(-1,-1,1\,\mid\,-1,1,1)
=
\eta_{++}\oplus\eta_{--}.
\]
Thus the pair-symmetric parity-invariant local constitutive sector contains
twelve independent algebraic variables.

Let \(G^{abcd}_{\rm par}\) be the corresponding area-metric tensor and let
\[
{\cal G}^{abcd}[G]=-
\frac{1}{24}\epsilon_{mnpq}\epsilon_{rstu}
G^{mnr(a}G^{b|ps|c}G^{d)qtu}
\]
be the Tamm--Rubilar tensor density.  For \(q_a=(\omega,x,y,z)\), define
\[
\widetilde P_{\rm par}(q)=
{\cal G}^{abcd}[G_{\rm par}]q_aq_bq_cq_d .
\]
If \(\Lambda_{\rm par}:={\cal G}^{0000}[G_{\rm par}]\neq0\), then the exact
monic characteristic polynomial is
\[
P_{\rm par}=\frac{\widetilde P_{\rm par}}{\Lambda_{\rm par}}
=\omega^4+C_3\omega^3+C_2\omega^2+C_1\omega+C_0,
\]
with
\[
C_3=\frac{4{\cal G}^{000i}k_i}{\Lambda_{\rm par}},\qquad
C_2=\frac{6{\cal G}^{00ij}k_ik_j}{\Lambda_{\rm par}},
\]
\[
C_1=\frac{4{\cal G}^{0ijk}k_ik_jk_k}{\Lambda_{\rm par}},\qquad
C_0=\frac{{\cal G}^{ijkl}k_ik_jk_kk_l}{\Lambda_{\rm par}},
\]
where \(k_i=(x,y,z)\).  Moreover,
\[
C_j(x,y,z)=C_j(x,y,-z),\qquad j=0,1,2,3 .
\]
\end{theorem}

\begin{proof}
The reflection \(P_\phi\) acts as \(+1\) on the bivectors \([01]\),
\([02]\), and \([12]\), and as \(-1\) on the bivectors \([03]\), \([13]\),
and \([23]\).  In the parity basis its representation is
\[
R_{\rm par}=\operatorname{diag}(+1,+1,+1,-1,-1,-1).
\]
A constitutive matrix inheriting the reflection symmetry satisfies
\(R_{\rm par}^{T}{\mathsf G}R_{\rm par}={\mathsf G}\).  Writing
\[
{\mathsf G}=\begin{pmatrix}G_{++}&G_{+-}\\G_{-+}&G_{--}\end{pmatrix},
\]
this condition gives \(G_{+-}=-G_{+-}\) and \(G_{-+}=-G_{-+}\), and
therefore \(G_{+-}=G_{-+}=0\).  The pair-exchange symmetry of the area
metric makes \(G_{++}\) and \(G_{--}\) symmetric \(3\times3\) matrices.
Hence the full local parity-invariant sector has twelve independent
algebraic variables.

For any area metric \(G^{abcd}\), the Tamm--Rubilar tensor density is cubic
in \(G\) and symmetric in its four free indices.  Therefore
\(\widetilde P_{\rm par}(q)\) is a homogeneous quartic polynomial in
\(q_a=(\omega,x,y,z)\).  Expanding the symmetric quartic according to the
number of temporal indices gives
\[
\widetilde P_{\rm par}=
{\cal G}^{0000}\omega^4
+4{\cal G}^{000i}\omega^3k_i
+6{\cal G}^{00ij}\omega^2k_ik_j
+4{\cal G}^{0ijk}\omega k_ik_jk_k
+{\cal G}^{ijkl}k_ik_jk_kk_l .
\]
If \(\Lambda_{\rm par}={\cal G}^{0000}\neq0\), division by
\(\Lambda_{\rm par}\) does not change the characteristic set.  This gives the
monic polynomial and the displayed coefficient formulas.

Finally, since \(G_{\rm par}\) is invariant under \(P_\phi\), the
Tamm--Rubilar tensor constructed cubically from it is also invariant.  The
transformed covector has \((\omega,x,y,z)\mapsto(\omega,x,y,-z)\).  Therefore
\(P_{\rm par}(\omega,x,y,z)=P_{\rm par}(\omega,x,y,-z)\).  By uniqueness of
the polynomial expansion in powers of \(\omega\), each coefficient \(C_j\) is
individually even in \(z\).
\end{proof}

\begin{remark}[SSSW-frame cubic cancellation]
\label{cor:sssw-cubic}
The parity theorem alone does not require \(C_3\) to vanish.  In the explicitly
reconstructed SSSW-frame subclass used in the supplementary implementation, the
magneto-electric trace condition imposed by the SSSW frame removes the
\(\omega^3\) Tamm--Rubilar coefficient.  This is verified by direct symbolic
Tamm--Rubilar contraction from the reconstructed tensor.  For that subclass, the
monic polynomial reduces to
\[
P_{\rm par}=\omega^4+C_2\omega^2+C_1\omega+C_0,
\]
with spatial coefficients taking the exact structural form
\[
C_2(x,y,z)=\alpha_0(x,y)+\alpha_2(x,y)z^2,
\]
\[
C_1(x,y,z)=\beta_0(x,y)+\beta_2(x,y)z^2,
\]
\[
C_0(x,y,z)=\gamma_0(x,y)+\gamma_2(x,y)z^2+\gamma_4(x,y)z^4.
\]
\end{remark}

The theorem uses only the reflection isometry and the Tamm--Rubilar tensor; it
is independent of any area-metric field equation or closure prescription.
SSSW-frame restrictions may remove \(C_3\), but this is not a consequence of
parity alone.  The full parity generator retained in the supplementary
implementation allows \(C_3\ne0\) and evaluates the twelve-variable polynomial
directly from the invariant formula.  The printed six-variable meridional
polynomial is the complementary analytic bridge: it keeps the
connection-motivated meridional block and supplies compact coefficient checks
for pointwise margin calculations.

The compact Tamm--Rubilar derivation is given in \cref{app:quartic-derivation}.
The fully expanded \(C_2,C_1,C_0\) coefficients are kept in the supplementary
symbolic framework rather than printed in the article; \cref{app:expanded-coefficients}
records the invariant sign checks used to audit that expansion.

\begin{table}[t]
\centering
\caption{Roles of the printed meridional polynomial and the full parity-invariant construction.}
\label{tab:psi-versus-parity}
\begin{tabular}{@{}lll@{}}
\toprule
Feature & Meridional \(P_\Psi\) & Full parity \(P_{\rm par}\) \\
\midrule
Algebraic data & six SSSW-frame variables & \(G_{++}\oplus G_{--}\) \\
Cubic term & absent by SSSW-frame condition & generally present \\
Transverse block & not explicit & included invariantly \\
Use & analytic sectoral benchmark & supplied-tensor roots \\
\bottomrule
\end{tabular}
\end{table}

Thus \(P_\Psi\) should not be used for a supplied tensor whose magneto-electric
trace data place it outside the reconstructed SSSW-frame subclass.  In that
case the appropriate local object is the full invariant polynomial
\(P_{\rm par}\), generally with \(C_3\ne0\).

In the restricted six-variable polynomial, every term in $C_1$ contains at least one of $a,b,c,d$, and hence at least one mixed variable.  Thus, within this meridional SSSW projection, the mixed sector is the exclusive source of odd-in-frequency root asymmetry.  Conversely, $C_2$ and $C_0$ are the actual algebraic objects that must be checked for sign and spatial definiteness; the signs of individual $\phi_A$ are not by themselves acceptance criteria.

To first order about the metric vacuum,
\begin{align}
P_\Psi={}&(\omega^2-x^2-y^2-z^2)^2 \\
&+(x^2+y^2+z^2-\omega^2)
\Bigl[(2\phi_1+\phi_{12}+\phi_4)x^2+\sqrt2\phi_2xy \\
&\hspace{3.1cm}+ (\phi_1+\phi_{12}+2\phi_4)y^2+(\phi_1+\phi_4)z^2\Bigr]+O(\phi^2).
\end{align}
In particular, $\phi_{16}$ and $\phi_{17}$ do not produce a first-order odd-in-frequency deformation about the exact metric vacuum unless coupled to a non-metric symmetric lift.  This is a useful guard against overinterpreting a single large mixed component.

\paragraph{Implementation checks.}
Generic full-parity supplied tensors, including cases with both \(G_{++}\) and \(G_{--}\) populated and \(C_3\neq0\), are evaluated by the supplementary symbolic implementation from the invariant formula above.  The printed six-variable polynomial remains the analytic bridge and compact audit case; the full generator is the object to use for tensors outside the reconstructed SSSW-frame subclass.

\section{Local control window and DH validity domain}
\label{sec:admissibility}

For each unit spatial covector \(k=(x,y,z)\), write
\begin{equation}
        p_{\Psi,k}(\Omega)=P_\Psi(\Omega,k),\qquad |k|_\delta=1,
\end{equation}
and compare with the metric double root \(p_{0,k}=(\Omega^2-1)^2\).  The
near-metric size of the reconstructed SSSW data is
\begin{equation}
 \epsilon_{\WB}:=\max\{\|U-I\|_2,\|L-I\|_2,\|W\|_2\}.
 \label{eq:wb-size}
\end{equation}
The reference value \(\epsilon_*=0.1\) is a diagnostic guardrail, not a universal
physical constant.  For a generic weak-birefringent tensor it keeps finite-amplitude
terms of schematic size \(O(\epsilon_{\WB}^2)\) at roughly the ten-percent level
relative to first-order shifts.  It would be misleading to claim that
\(\epsilon_{\WB}=0.1\) is where quadratic terms rival the linear terms; rivaling
occurs parametrically at \(\epsilon_{\WB}=O(1)\).  For DH input the stronger EFT
conditions below, not \(\epsilon_*\), determine the physical range of the
low-frequency surrogate.

Define the symmetric positivity, coefficient, branch, and cone-displacement
margins by
\begin{align}
 m_{\rm sym}&=\min\left\{\mu_{\min}\begin{pmatrix}u&s\\s&w\end{pmatrix},\ell,uw-s^2\right\},
        \label{eq:msym}\\
 m_2&=\inf_{|k|=1}\left[-\frac12 C_2(k)\right],
 &m_0&=\inf_{|k|=1}C_0(k),\label{eq:m2m0}
\end{align}
with real roots ordered, when possible, as
\begin{equation}
 \Omega_-^{(2)}\le \Omega_-^{(1)}<0<\Omega_+^{(1)}\le\Omega_+^{(2)} .
\end{equation}
Then
\begin{equation}
 m_E=\inf_{|k|=1}\min\{\Omega_+^{(1)}(k),-\Omega_-^{(1)}(k)\},
\end{equation}
\begin{equation}
 d_{\rm cone}=\sup_{|k|=1}\max\left\{
 \max_{\sigma=1,2}|\Omega_+^{(\sigma)}(k)-1|,
 \max_{\sigma=1,2}|\Omega_-^{(\sigma)}(k)+1|
 \right\} .
\end{equation}
For \(0<\epsilon\le0.2\), the local acceptance window is
\begin{equation}
\begin{aligned}
{\cal W}^{\rm loc}_\epsilon:=\{\Psi:\;&
 p_{\Psi,k}(\Omega) \text{ has four real roots for all } |k|=1,\\
&\text{with ADM-normal positive/negative separation,}\quad
\epsilon_{\WB}\le\epsilon,\\
&m_{\rm sym},m_2,m_0,m_E\ge 1-\epsilon,
\qquad d_{\rm cone}\le\epsilon\} .
\end{aligned}
\label{eq:admissible-set-epsilon}
\end{equation}
Membership in \({\cal W}^{\rm loc}_\epsilon\) is used as a local weak-birefringence
acceptance window for a supplied nondispersive tensor.  It controls the root data
that a ray integrator actually consumes, while global hyperbolicity, stable
causality, dual-cone regularity, and dynamical well-posedness remain separate
analyses of the underlying theory.

\subsection{Low-frequency Fresnel roots versus QED wavefront characteristics}
\label{subsec:low-frequency-wavefronts}

The tests above concern a supplied nondispersive Tamm--Rubilar polynomial.  For a
genuine nondispersive medium, that polynomial is the local Fresnel polynomial.
For DH input it is only the low-frequency algebraic curvature surrogate.  The
full QED problem contains derivative and nonlocal vacuum-polarization terms, and
its high-frequency wavefront limit is governed by the frequency-dependent
polarization operator, not by the local quartic used here
\cite{HollowoodShore2007,HollowoodShore2008,DebDesaiGhosh2026}.

This limited status is also the reason the surrogate is useful.  In a numerical
pipeline one often needs to decide, at millions of local evaluation points,
whether a supplied effective tensor has real and continuously labelable branches
before performing an expensive Hamiltonian integration.  The Tamm--Rubilar screen
is therefore a preprocessing filter: accepted points are branch-stable inputs to
a low-frequency ray calculation, while rejected points are flagged as
``nondispersive surrogate not validated'' and should be treated with the
frequency-dependent polarization operator rather than forced through a local
quartic root finder.

\subsection{Dispersive truncation bound for numerical use}
\label{subsec:qed-dispersion-error}

Let \(\omega_{\rm loc}\) be the local photon frequency,
\({\cal R}(x)=\max_{\hat a\hat b\hat c\hat d}|R_{\hat a\hat b\hat c\hat d}|\), and
let \(L_R\) be a local curvature-variation scale.  The nondispersive DH screen is
controlled only when
\begin{equation}
\epsilon_{\rm curv}=\lambda_e^2{\cal R}\ll1,
\quad
\epsilon_{\rm grad}=\frac{\lambda_e}{L_R}\ll1,
\quad
\epsilon_{\rm eik}=\frac{1}{\omega_{\rm loc}L_R}\ll1,
\quad
\epsilon_{\rm freq}=\frac{\omega_{\rm loc}}{m_e}\ll1 .
\label{eq:qed-validity-parameters}
\end{equation}
The leading DH cone shift scales as
\begin{equation}
        \delta_{\rm DH}=O\!\left(\frac{\alpha}{\pi}\lambda_e^2{\cal R}\right),
\end{equation}
while the omitted dispersive, gradient, and higher-curvature contributions have
the schematic relative size
\begin{equation}
 \eta_{\rm EFT}:=
 \left(\frac{\omega_{\rm loc}}{m_e}\right)^2
 +\frac{\lambda_e}{L_R}
 +\lambda_e^2{\cal R},
 \qquad
 \frac{|\Delta\delta_{\rm DH}|}{|\delta_{\rm DH}|}
 \lesssim C_{\rm disp}\eta_{\rm EFT} .
\label{eq:qed-relative-error-bound}
\end{equation}
Thus omitted terms rival the first-order cone shift when
\(C_{\rm disp}\eta_{\rm EFT}=O(1)\).  A conservative numerical use of the local
DH surrogate should instead require \(C_{\rm disp}\eta_{\rm EFT}\lesssim0.1\), in
parallel with the algebraic root window \({\cal W}^{\rm loc}_{\epsilon_*}\).  The
Hollowood--Shore dispersive parameter may be written schematically as
\begin{equation}
        \widehat\omega^{\,2}\sim\frac{\omega_{\rm loc}^2{\cal R}}{m_e^4},
\end{equation}
and when this parameter or \(\omega_{\rm loc}/m_e\) ceases to be small, a code
should return ``nondispersive surrogate not validated'' rather than infer a
microscopic causality failure.

\begin{figure}[t]
\centering
\begin{tikzpicture}[>=Stealth]
\scriptsize
\tikzset{box/.style={draw,rounded corners,align=center,text width=2.45cm,minimum height=0.72cm}}
\node[box] (sup) at (-4.5,2.0) {supplied local\\tensor $G^{abcd}(x)$};
\node[box] (tr)  at (0,2.0) {Tamm--Rubilar\\roots and margins};
\node[box] (eft) at (4.5,2.0) {QED/EFT\\status marker};
\node[box] (wg)  at (0,0.55) {$\mathcal W^{\rm loc}_{\epsilon_*}$ gate\\real roots, branch gap\\$\epsilon_*=0.1$ guardrail};
\node[box] (eg)  at (4.5,0.55) {EFT gate\\$C_{\rm disp}\eta_{\rm EFT}\le0.1$\\low frequency/eikonal};
\node[box] (acc) at (0,-1.45) {accepted branch\\Hamiltonian transport};
\node[box,dashed] (pol) at (4.5,-1.45) {replace by dispersive\\polarization determinant};
\node[align=center,text width=3.5cm] (fail) at (-4.5,-1.45) {if either gate fails:\\return local-surrogate\\status flag, not QED acausality};
\draw[->,thick] (sup)--(tr);
\draw[->,thick] (tr)--(eft);
\draw[->,thick] (tr)--(wg);
\draw[->,thick] (eft)--(eg);
\draw[->,thick] (wg)--(acc);
\draw[->,thick] (eg)--(pol);
\end{tikzpicture}
\caption{Numerical use of the nondispersive surrogate.  The algebraic window
$\mathcal W^{\rm loc}_{\epsilon_*}$ controls root reality and branch stability;
the separate EFT marker controls whether the DH low-frequency approximation is
valid.  Both gates must pass before a polarization ray code treats the local
quartic branches as transport data.}
\label{fig:surrogate-workflow}
\end{figure}

\section{Physical application: exterior Schwarzschild Drummond--Hathrell tensor}
\label{sec:dh}

The Drummond--Hathrell effective action supplies a physically fixed local
constitutive tensor in the low-frequency electron-loop regime
\cite{DrummondHathrell}.  In this section it is used as a real full-parity test
of \cref{thm:complete-parity-polynomial}.  Only the local algebraic
\(O(\alpha R/m_e^2)\) curvature sector is inserted into the Tamm--Rubilar
construction; the higher-derivative DH term belongs to the separate dispersive
QED wavefront problem discussed in \cref{subsec:qed-dispersion-error} and in
\cite{HollowoodShore2007,HollowoodShore2008}.  Thus the roots below are
low-frequency phase-velocity roots of the supplied nondispersive surrogate, not
microscopic high-frequency QED characteristics.

With metric signature \((-+++)\), antisymmetrization
\(X^{a[c}Y^{d]b}=(X^{ac}Y^{db}-X^{ad}Y^{cb})/2\), and constitutive convention
\(H^{ab}=(1/2)\chi^{abcd}F_{cd}\), the local algebraic DH perturbation is
\begin{equation}
 \delta\chi^{abcd}_{\rm DH}
 =-\frac{8}{m_e^2}\left[
 c_RR\,g^{a[c}g^{d]b}
 +\frac12c_{\rm Ric}\left(R^{a[c}g^{d]b}-R^{b[c}g^{d]a}\right)
 +c_{\rm Riem}R^{abcd}\right],
 \label{eq:dh-chi}
\end{equation}
where
\begin{equation}
 c_R=-\frac{1}{144}\frac{\alpha}{\pi},\qquad
 c_{\rm Ric}=\frac{13}{360}\frac{\alpha}{\pi},\qquad
 c_{\rm Riem}=-\frac{1}{360}\frac{\alpha}{\pi} .
 \label{eq:dh-coeffs}
\end{equation}
The sign and normalization dictionary from \(\delta\chi^{abcd}\) to the
restricted meridional SSSW variables, together with the complementary transverse
parity block, is recorded in \cref{app:dictionary}.  The main-text calculation
below uses the full bivector tensor, not only the restricted meridional block.

\paragraph{Exterior Schwarzschild full-parity tensor.}
In the Ricci-flat exterior Schwarzschild geometry, use the adapted orthonormal
basis \((0,1,2,3)=(n,\hat r,\hat\theta,\hat\phi)\) and define
\(\kappa=M/r^3\).  With the curvature convention of \eqref{eq:dh-chi}, the
nonzero independent raised Riemann components are
\[
R^{0101}=-2\kappa,\qquad
R^{0202}=R^{0303}=\kappa,
\]
\[
R^{1212}=R^{1313}=-\kappa,
\qquad
R^{2323}=2\kappa .
\]
Since \(R_{ab}=R=0\), \eqref{eq:dh-chi} becomes
\begin{equation}
 \delta\chi^{abcd}_{\rm DH}
   =\frac{2\alpha}{90\pi}\lambda_e^2 R^{abcd} .
 \label{eq:schw-dh-ricci-flat}
\end{equation}
Writing
\begin{equation}
 \epsilon_{\rm S}:=\frac{\alpha}{45\pi}\lambda_e^2\frac{M}{r^3},
 \label{eq:epsS-def}
\end{equation}
the supplied local tensor has both parity blocks fixed:
\begin{equation}
G^{\rm Schw}_{++}
 =\diag(-1-2\epsilon_{\rm S},\,-1+\epsilon_{\rm S},\,1-\epsilon_{\rm S}),
\qquad
G^{\rm Schw}_{--}
 =\diag(-1+\epsilon_{\rm S},\,1-\epsilon_{\rm S},\,1+2\epsilon_{\rm S}).
 \label{eq:schw-full-blocks}
\end{equation}
Applying the full invariant Tamm--Rubilar contraction gives
\begin{equation}
\begin{aligned}
P^{\rm Schw}_{\rm DH}={}&
\left[\omega^2-x^2-
 \frac{1-\epsilon_{\rm S}}{1+2\epsilon_{\rm S}}(y^2+z^2)\right]\\
&\times
\left[\omega^2-x^2-
 \frac{1+2\epsilon_{\rm S}}{1-\epsilon_{\rm S}}(y^2+z^2)\right].
\end{aligned}
\label{eq:schw-full-parity-polynomial}
\end{equation}
The exact rational dependence in \eqref{eq:schw-full-parity-polynomial} is the
polynomial of the supplied nondispersive surrogate.  The physical DH comparison
is its first-order expansion,
\begin{equation}
 Q_\pm=\omega^2-x^2-
 \left[1\mp3\epsilon_{\rm S}+O(\epsilon_{\rm S}^2)\right](y^2+z^2).
 \label{eq:schw-quadratic-factors}
\end{equation}
Radial propagation has \(y=z=0\) and is unshifted.  For orbital propagation,
\(x=0\), the two positive phase roots are
\begin{equation}
 \Omega_{\rm slow,fast}
 =1\mp\frac{3}{2}\epsilon_{\rm S}+O(\epsilon_{\rm S}^2)
 =1\mp\frac{\alpha}{30\pi}\lambda_e^2\frac{M}{r^3}
  +O\!\left(\lambda_e^4\frac{M^2}{r^6}\right).
 \label{eq:schw-DH-known-limit}
\end{equation}
This is the standard low-frequency Drummond--Hathrell Schwarzschild result:
radial photons are unshifted, while orbital photons split into equal-and-opposite
polarization branches at order \(\alpha\lambda_e^2M/(\pi r^3)\)
\cite{DrummondHathrell,HollowoodShore2007,HollowoodShore2008}.  The original
DH reference is therefore cited at the point where the known limit is used.

The exact factorization also certifies the local root margins.  For
\(0<\epsilon_{\rm S}<1\) and \(|k|=1\), set
\[
 A_{\rm S}=\frac{1-\epsilon_{\rm S}}{1+2\epsilon_{\rm S}},\qquad
 B_{\rm S}=\frac{1+2\epsilon_{\rm S}}{1-\epsilon_{\rm S}}=A_{\rm S}^{-1}.
\]
The roots are exactly
\(\Omega=\pm\sqrt{x^2+A_{\rm S}(y^2+z^2)}\) and
\(\Omega=\pm\sqrt{x^2+B_{\rm S}(y^2+z^2)}\), hence
\[
 m_E^{\rm Schw}=\sqrt{A_{\rm S}},\qquad
 d_{\rm cone}^{\rm Schw}
 =\max\left\{1-\sqrt{A_{\rm S}},\sqrt{B_{\rm S}}-1\right\}.
\]
The extrema occur at tangential directions \(x=0\).  These are closed-form
bounds, not sampled minima.

For any static nonrotating compact-object background with vanishing shift,
\(K_{ij}=0\) gives \(R_{\hat0\hat i}=0\) and
\(R_{\hat0\hat i\hat j\hat k}=0\).  The restricted mixed meridional DH slots
therefore vanish in the adapted tetrad.  Rotating or nonstatic systems should be
handled with the full local pair-symmetric tensor rather than by forcing them
into the restricted six-variable diagnostic.

\section{First-order Ricci-flat screen reduction and a Kerr benchmark application}
\label{sec:ricci-flat-propagation}

The Schwarzschild calculation in \cref{sec:dh} is an exact Tamm--Rubilar
factorization for the supplied nondispersive surrogate.  For propagation
applications, the more useful object is the corresponding first-order branch
Hamiltonian along a ray.  This section gives a Ricci-flat screen reduction and
two applications: a ray-integrated Schwarzschild birefringent delay and a Kerr
weak-lensing retardance calculation.  Throughout this section
\begin{equation}
 \zeta_{\rm DH}:=\frac{\alpha}{45\pi}\lambda_e^2 .
 \label{eq:zeta-DH}
\end{equation}
Only terms through first order in \(\zeta_{\rm DH}R\) are interpreted as
Drummond--Hathrell QED; terms of order \((\zeta_{\rm DH}R)^2\) belong to the
finite-amplitude algebraic surrogate, not to higher-loop QED.

\begin{proposition}[Ricci-flat Drummond--Hathrell screen reduction and branch margins]
\label{thm:ricci-flat-screen}
Let \((M,g)\) be Ricci-flat on an open region \(U\), and insert only the
algebraic low-frequency Drummond--Hathrell curvature term into the local
constitutive tensor.  At a point choose an orthonormal tetrad, write a nonzero
spatial covector as \(k_i=|k|n_i\), \(n_in^i=1\), and set
\(\ell^a=(1,n^i)\).  Let \(s_A^a=(0,s_A^i)\), \(A=1,2\), be an orthonormal screen
basis with \(s_A\cdot n=0\), and define
\begin{equation}
 {\cal K}_{AB}(n):=R_{abcd}\ell^a s_A^b\ell^c s_B^d .
 \label{eq:screen-K}
\end{equation}
If \(\lambda_A(n)\) are the eigenvalues of this real symmetric screen matrix,
then the two physical Drummond--Hathrell branches are, through first order,
\begin{equation}
 H_A(q)=\omega^2-|k|^2-\zeta_{\rm DH}|k|^2\lambda_A(n),
 \qquad A=1,2,
 \label{eq:screen-branches}
\end{equation}
with positive phase roots
\begin{equation}
 \Omega_A(n):=\frac{\omega_A}{|k|}
 =1+\frac{\zeta_{\rm DH}}{2}\lambda_A(n)
 +O\!\left((\zeta_{\rm DH}R)^2\right).
 \label{eq:screen-roots}
\end{equation}
Equivalently, after diagonalizing the first-order screen operator,
\begin{equation}
 \det M=\prod_{A=1}^{2}
 \left[\omega^2-|k|^2-\zeta_{\rm DH}|k|^2\lambda_A(n)\right]
 +O\!\left((\zeta_{\rm DH}R)^2|k|^4\right).
 \label{eq:screen-factorization}
\end{equation}
Equation \eqref{eq:screen-factorization} is a determinant of the diagonalized
first-order polarization operator.  It should not be read as a linearly
truncated quartic polynomial: in Ricci-flat backgrounds
\(\lambda_1+\lambda_2=0\), so the term linear in \(\zeta_{\rm DH}\) cancels in
the quartic even though the double metric root has split at first order.  This
cancellation is the elementary identity
\(\prod_A[X-\zeta_{\rm DH}\lambda_A]=X^2-
\zeta_{\rm DH}^2\lambda_+^2\) for a trace-free two-by-two screen matrix; the
useful content of the proposition is the covariant screen operator, the branch
Hamiltonians, and the explicit local root-margin bounds.

If
\begin{equation}
 \zeta_{\rm DH}\Lambda_U<1,
 \qquad
 \Lambda_U:=\sup_{x\in U}\sup_{|n|=1}\|{\cal K}_x(n)\|_2,
 \label{eq:screen-margin-condition}
\end{equation}
then every branch Hamiltonian in \eqref{eq:screen-branches} has real roots with
the same ADM-normal positive/negative orientation throughout \(U\), and the local
branch margins obey
\begin{equation}
 m_E\ge \sqrt{1-\zeta_{\rm DH}\Lambda_U},
 \qquad
 d_{\rm cone}\le
 \max\!\left\{\sqrt{1+\zeta_{\rm DH}\Lambda_U}-1,
 1-\sqrt{1-\zeta_{\rm DH}\Lambda_U}\right\}.
 \label{eq:screen-margin-bounds}
\end{equation}
This is the local first-order branch-realness and root-margin certificate needed
before the accepted branches are passed to a transport calculation; the full
frequency-dependent QED operator is handled by the separate EFT status marker of
\cref{subsec:qed-dispersion-error}.
\end{proposition}

\begin{proof}
For an eikonal field \(F_{ab}=2q_{[a}a_{b]}\), the algebraic Maxwell equation
\(q_aH^{ab}=0\), with \(q\cdot a=0\), projects onto the two physical screen
polarizations.  In Ricci-flat spacetime the DH constitutive perturbation is
\(\delta\chi^{abcd}_{\rm DH}=\zeta_{\rm DH}R^{abcd}\), so the screen-projected
operator takes the form
\begin{equation}
 M_{AB}=\left(\omega^2-|k|^2\right)\delta_{AB}
        -\zeta_{\rm DH}|k|^2{\cal K}_{AB}(n)
        +O\!\left((\zeta_{\rm DH}R)^2|k|^2\right).
 \label{eq:screen-matrix-proof}
\end{equation}
Terms proportional to \(q_A\) or to the gauge direction drop out of the screen
projection, and the curvature term may be evaluated on the unperturbed null
vector \(\ell^a\) because the root displacement is already first order.  The
Riemann symmetries make \({\cal K}_{AB}\) symmetric, so an orthogonal screen
rotation diagonalizes \eqref{eq:screen-matrix-proof}.  The diagonal entries give
\eqref{eq:screen-branches}; solving them perturbatively gives
\eqref{eq:screen-roots}; and their product gives \eqref{eq:screen-factorization}.
The margin bounds follow immediately from
\(|\lambda_A|\le\Lambda_U\).  The trace-free Ricci-flat identity follows by
contracting \({\cal K}_{AB}\) over the screen with
\(g^{ab}+\ell^{(a}N^{b)}\), where \(N\) is the auxiliary null vector satisfying
\(\ell\cdot N=-2\), and using \(R_{ab}=0\).
\end{proof}

For Ricci-flat backgrounds the trace of \({\cal K}_{AB}\) over the two-dimensional
screen vanishes, so the two eigenvalues are \(\lambda_\pm=\pm\Lambda(n)\).  The
low-frequency branch Hamiltonians may therefore be written, to first order in
the local orthonormal tetrad, as
\begin{equation}
 H_\pm(x,q)=\omega^2-|k|^2
 -\zeta_{\rm DH}|k|^2\lambda_\pm(x,n).
 \label{eq:branch-Hamiltonians}
\end{equation}
The corresponding ray calculation is standard Hamiltonian transport for each
accepted branch.  The observable used below is the arrival-time difference at
fixed Killing frequency in a stationary, asymptotically flat lensing geometry.
For the nondispersive first-order surrogate the phase and group velocities agree
to this order.  Since \(\omega_A=|k|[1+\zeta_{\rm DH}\lambda_A/2]\), a branch with
eigenvalue \(\lambda_A\) has
\begin{equation}
        v_{g,A}=1+\frac{\zeta_{\rm DH}}{2}\lambda_A+O((\zeta_{\rm DH}R)^2),
        \qquad
        t_A=\int_\gamma\left[1-\frac{\zeta_{\rm DH}}{2}\lambda_A(x,n)\right]\dd l .
        \label{eq:branch-time-from-Hamiltonian}
\end{equation}
Thus the branch-labelled split of branch \(B\) relative to branch \(A\) is
\begin{equation}
        \Delta t_{AB}:=t_B-t_A=
        \frac{\zeta_{\rm DH}}{2}\int_\gamma
        [\lambda_A(x,n)-\lambda_B(x,n)]\,\dd l
        +O((\zeta_{\rm DH}R)^2),
        \label{eq:signed-time-delay-general}
\end{equation}
and the corresponding positive unlabeled magnitude is
\begin{equation}
 |\Delta t_{12}|=
 \frac{\zeta_{\rm DH}}{2}\int_\gamma
 |\lambda_1(x,n)-\lambda_2(x,n)|\,\dd l
 =\zeta_{\rm DH}\int_\gamma \Lambda(x,n)\,\dd l .
 \label{eq:time-delay-general}
\end{equation}
For a monochromatic beam the same quantity is equivalently a phase retardance,
\(\delta_{AB}=\omega_\infty\Delta t_{AB}\), acting on a specified input Jones or
Stokes vector.  It is not polarization production from an unpolarized beam and
must be used together with the QED validity conditions in
\cref{subsec:qed-dispersion-error}.

\paragraph{Check against Schwarzschild.}
For Schwarzschild, \(R_{0i0j}=\diag(-2\kappa,\kappa,\kappa)\),
\(\kappa=M/r^3\), in the adapted orthonormal frame of \cref{sec:dh}.  Radial
propagation gives \({\cal K}=0\).  Tangential propagation gives
\(\lambda_\pm=\pm3\kappa\), so \eqref{eq:screen-roots} gives
\(\Omega_\pm=1\pm(3/2)\epsilon_{\rm S}\), which is exactly the first-order form
of \eqref{eq:schw-DH-known-limit}, up to the labeling of the slow and fast
polarizations.

\paragraph{Weak-lensing time-delay split in Schwarzschild.}
The local factorization can now be turned into a propagation observable.  In the
weak-deflection regime \(b\gg M\), approximate the unperturbed metric ray by a
straight line with impact parameter \(b\), affine Euclidean coordinate \(l\), and
\(r(l)=\sqrt{b^2+l^2}\).  The angle between the ray direction and the radial unit
vector satisfies \(\sin\psi=b/r\).  The Schwarzschild screen eigenvalues are
therefore
\begin{equation}
 \lambda_\pm(l)=\pm 3\frac{M}{r(l)^3}\sin^2\psi
 =\pm 3M\frac{b^2}{\left(b^2+l^2\right)^{5/2}} .
 \label{eq:schw-grazing-lambda}
\end{equation}
Equation \eqref{eq:time-delay-general} gives the ray-integrated split
\begin{equation}
 |\Delta t_{12}^{\rm Schw}|
 =3\zeta_{\rm DH}Mb^2
   \int_{-\infty}^{\infty}\frac{\dd l}{(b^2+l^2)^{5/2}}
 =\frac{4\alpha}{45\pi}\lambda_e^2\frac{M}{b^2}
 +O\!\left(\zeta_{\rm DH}\frac{M^2}{b^3},(\zeta_{\rm DH}R)^2\right).
 \label{eq:schw-grazing-delay}
\end{equation}
This calculation shows how the certified local branches produce a concrete
birefringent propagation quantity once an unperturbed ray, a branch labelling,
and the QED validity window have been specified.

\subsection{Slow Kerr: principal-frame magnetic-Weyl tilt of the DH screen eigenbasis}
\label{subsec:slow-kerr}

Daniels and Shore studied the same Drummond--Hathrell curvature coupling in
Kerr, emphasizing velocity shifts, the polarization sum rule, the horizon
theorem, and the stationary-limit surface \cite{DanielsShore1996}.  The present
benchmark starts from that local DH setting but fixes a different observable
problem: a two-ended weak-lensing experiment with an asymptotic Sachs screen,
branch labelling, and a phase/group-delay retardance at fixed
\(\omega_\infty\).  Locally, in a specified principal orthonormal frame, the
magnetic Weyl tensor tilts the eigenlines of the DH screen operator; the Born
calculation below then shows which part of this local data survives as an
asymptotic scattering quantity.

In a principal orthonormal frame the Kerr complex tidal matrix is
\cite{Kerr1963,Chandrasekhar1983}
\begin{equation}
 E_{ij}+iB_{ij}
 =\frac{M}{(r-i a\cos\theta)^3}\diag(-2,1,1),
 \label{eq:kerr-complex-tidal}
\end{equation}
up to the sign of the imaginary part under reversal of the spatial triad.  To
first order in \(a/r\), define \(\kappa=M/r^3\) and
\(\beta_K=3a\cos\theta/r\).  For azimuthal propagation \(n=e_{\hat\phi}\), with
the principal screen basis \((e_{\hat r},e_{\hat\theta})\), the branch screen is
\begin{equation}
 {\cal K}_{\hat\phi}
 =3\kappa
 \begin{pmatrix}
 -1 & \beta_K\\
 \beta_K & 1
 \end{pmatrix}
 +O\!\left(\kappa\frac{a^2}{r^2}\right).
 \label{eq:kerr-azimuthal-screen}
\end{equation}
The eigenvalues are therefore
\begin{equation}
 \lambda_\pm^{\rm Kerr}=\pm3\kappa
 \sqrt{1+\beta_K^2}
 =\pm3\kappa\left[1+O\!\left(\frac{a^2}{r^2}\right)\right],
 \label{eq:kerr-eigenvalues}
\end{equation}
so the azimuthal phase-velocity split in this principal-frame screen has no term
linear in spin.  The diagonalizing angle, however, is spin-odd:
\begin{equation}
 \tan(2\chi_K)=\frac{2{\cal K}_{\hat r\hat\theta}}
 {{\cal K}_{\hat r\hat r}-{\cal K}_{\hat\theta\hat\theta}}
 =-\beta_K,
 \qquad
 \chi_K=-\frac{3a\cos\theta}{2r}+O\!\left(\frac{a^3}{r^3}\right).
 \label{eq:kerr-polarization-rotation}
\end{equation}
Thus the magnetic Weyl part acts as a shear of the local screen: the electric
Weyl tensor fixes the two Schwarzschild-like principal axes, while the
spin-induced magnetic Weyl entry tilts the DH eigenlines by half the ratio of
off-diagonal magnetic shear to electric tidal anisotropy.

The angle \(\chi_K\) is a local orientation: it measures the two DH eigenlines
relative to the specified principal orthonormal Kerr screen
\((e_{\hat r},e_{\hat\theta})\).  Under a local \(SO(2)\) rotation of the screen
basis by \(\psi(x)\),
\begin{equation}
        s_A\mapsto R_A{}^B(\psi)s_B,
        \qquad \chi_K\mapsto\chi_K-\psi,
\end{equation}
while the eigenvalues \(\lambda_\pm\) are unchanged.  Equation
\eqref{eq:kerr-polarization-rotation} is therefore the principal-frame local
eigenbasis tilt whose transport-covariant comparison is defined next.

For an operational comparison along a null ray \(\gamma(\lambda)\), choose a
Sachs screen basis \((s_1,s_2)\) and define its screen connection along the ray
by
\begin{equation}
        \varpi_\lambda:=s_2\cdot\nabla_{\dot\gamma}s_1 .
        \label{eq:sachs-screen-connection}
\end{equation}
If \(\chi(\lambda)\) denotes the local DH eigenbasis angle measured in this
chosen screen, then the accumulated eigenaxis mismatch between two points is the
invariant combination
\begin{equation}
        \Theta_\gamma=
        \chi(\lambda_2)-\chi(\lambda_1)
        +\int_{\lambda_1}^{\lambda_2}\varpi_\lambda\,\dd\lambda .
        \label{eq:transport-covariant-eigenaxis-mismatch}
\end{equation}
Indeed, under \(s_A\mapsto R_A{}^B(\psi)s_B\), one has
\(\varpi_\lambda\mapsto\varpi_\lambda+\dot\psi\), so
\eqref{eq:transport-covariant-eigenaxis-mismatch} is unchanged.  In a
parallel-transported Sachs basis \(\varpi_\lambda=0\), but then the diagonalizing
angle must be recomputed in that transported basis rather than read directly
from the local principal tetrad.

Relative to the same principal tetrad, two nearby azimuthal ray points have the
local scalar change
\begin{equation}
 \Delta\chi_K
 =-\frac{3a}{2}\left(\frac{\cos\theta_2}{r_2}
       -\frac{\cos\theta_1}{r_1}\right)
 +O\!\left(\frac{a^3}{r^3}\right),
 \label{eq:kerr-axis-mismatch}
\end{equation}
and equivalently
\begin{equation}
 \dd\chi_K=\frac{3a}{2}\left(\frac{\cos\theta}{r^2}\,\dd r
 +\frac{\sin\theta}{r}\,\dd\theta\right)+O\!\left(\frac{a^3}{r^3}\right).
 \label{eq:kerr-axis-one-form}
\end{equation}
A polarization-resolved ray-tracing code should therefore evolve the
transport-covariant eigenbasis rate
\begin{equation}
        \dot\Theta_K=\dot\chi_K+\varpi_\lambda,
        \label{eq:kerr-transport-rate}
\end{equation}
where \(\varpi_\lambda=s_2\cdot\nabla_{\dot\gamma}s_1\) is the connection of the
chosen screen frame.  Equation \eqref{eq:kerr-axis-one-form} is only the
principal-frame variation of the local algebraic tilt.  The periastron check
below proves explicitly that the spin-odd Born eigenvalue term is a
longitudinal null-normalization effect, not an \(SO(2)\) screen-rotation
paradox.

\paragraph{Asymptotic grazing-ray benchmark.}
For an asymptotic source--observer experiment the endpoint curvature vanishes,
so a local DH axis becomes degenerate at both ends and an observed polarization
angle must be defined by transport from the finite interaction region to
asymptotic Sachs frames.  A perfectly unpolarized incident beam is unchanged by a
unitary birefringent retardance; the scattering quantity below is therefore a
retardance, or the associated Stokes-vector change, for a specified polarized
input.

Take the standard weak-lensing geometry in asymptotically Cartesian coordinates.
The metric ray is approximated by
\begin{equation}
        \bm x(l)=l\,\hat{\bm k}+b\,\hat{\bm b},\qquad
        \hat{\bm k}\cdot\hat{\bm b}=0,\qquad
        r(l)=\sqrt{b^2+l^2},
        \label{eq:kerr-born-ray}
\end{equation}
where \(\hat{\bm k}\) is the incoming-to-outgoing zeroth-order propagation
direction, \(b\hat{\bm b}\) is the impact vector, and
\(\hat{\bm p}:=\hat{\bm k}\times\hat{\bm b}\) is the oriented normal to the lens
plane.  The tangent \(L^\mu=(1,\hat{\bm k})\) is a Minkowski-null Born vector.  The
exact metric null tangent differs from it by
\(\delta L=O(M/r,J/r^2)\), so inserting \(L\) into the first-order Riemann
projection changes the screen only by
\(O(M^2/r^4,MJ/r^5,J^2/r^6)\), which is part of the displayed remainder.  The
integrated retardance below is normalized at infinity; pointwise screen entries
are intermediate quantities in this chosen asymptotic embedding.

Write the Kerr angular momentum as
\begin{equation}
        \bm J=J_b\hat{\bm b}+J_p\hat{\bm p}+J_k\hat{\bm k},
        \qquad J_p=\bm J\cdot(\hat{\bm k}\times\hat{\bm b}).
        \label{eq:kerr-spin-components}
\end{equation}
For reproducibility, use the linearized Kerr perturbation, in the curvature sign
convention of \eqref{eq:screen-K},
\begin{equation}
        h_{00}=\frac{2M}{r},\qquad
        h_{0i}=-\frac{2(\bm J\times\bm x)_i}{r^3},\qquad
        h_{ij}=\frac{2M}{r}\delta_{ij}.
        \label{eq:linearized-kerr-perturbation}
\end{equation}
The only curvature input is the linearized Riemann tensor
\begin{equation}
        R_{\mu\nu\rho\sigma}^{(1)}=\frac12
        \left(h_{\mu\sigma,\nu\rho}+h_{\nu\rho,\mu\sigma}
        -h_{\mu\rho,\nu\sigma}-h_{\nu\sigma,\mu\rho}\right),
        \label{eq:linearized-riemann}
\end{equation}
and the Born projection
\begin{equation}
        {\cal K}_{AB}=R_{0A0B}+2\hat k^iR_{0AiB}
        +\hat k^i\hat k^jR_{iAjB},\qquad A,B\in\{b,p\}.
        \label{eq:kerr-born-projection-formula}
\end{equation}
To evaluate it, choose axes
\((\hat{\bm b},\hat{\bm p},\hat{\bm k})=(e_1,e_2,e_3)\),
\(\bm x=(b,0,l)\), and \(\bm J=(J_b,J_p,J_k)\); rotational covariance then
restores the vector notation.  The derivative identities used are
\begin{align}
        \partial_i\partial_j r^{-1}
        &=\frac{3x_ix_j-r^2\delta_{ij}}{r^5},\nonumber\\
        \partial_j A_i
        &=\epsilon_{imn}J_m\left(\frac{\delta_{nj}}{r^3}
        -\frac{3x_nx_j}{r^5}\right),\nonumber\\
        \partial_j\partial_k A_i
        &=\epsilon_{imn}J_m\left[-\frac{3(\delta_{nj}x_k+\delta_{nk}x_j+\delta_{jk}x_n)}{r^5}
        +\frac{15x_nx_jx_k}{r^7}\right],
        \qquad A_i=\frac{(\bm J\times\bm x)_i}{r^3}.
        \label{eq:kerr-derivative-identities}
\end{align}
Substitution in \eqref{eq:kerr-born-projection-formula} gives the mass entries
\({\cal K}_{bb}^{(M)}=-3Mb^2/r^5\),
\({\cal K}_{pp}^{(M)}=+3Mb^2/r^5\),
\({\cal K}_{bp}^{(M)}=0\), and the spin entries
\begin{equation}
        {\cal K}_{bb}^{(J)}=-\frac{6J_p b}{r^5},\qquad
        {\cal K}_{pp}^{(J)}=+\frac{6J_p b}{r^5},\qquad
        {\cal K}_{bp}^{(J)}=-\frac{3b[J_b(3b^2-2l^2)+5J_kbl]}{r^7}.
        \label{eq:kerr-spin-screen-components}
\end{equation}
Therefore the first-order screen matrix in the parallel asymptotic screen
\((\hat{\bm b},\hat{\bm p})\) is
\begin{equation}
 {\cal K}_{\rm Born}(l)=
 \begin{pmatrix}
 -\dfrac{3Mb^2}{r^5}-\dfrac{6J_p b}{r^5} &
 -\dfrac{3b\,[J_b(3b^2-2l^2)+5J_kbl]}{r^7}\\[1.1em]
 -\dfrac{3b\,[J_b(3b^2-2l^2)+5J_kbl]}{r^7} &
 \dfrac{3Mb^2}{r^5}+\dfrac{6J_p b}{r^5}
 \end{pmatrix}
 +O\!\left(\frac{M^2}{r^4},\frac{MJ}{r^5},\frac{J^2}{r^6}\right).
        \label{eq:kerr-born-screen}
\end{equation}
This is a transported, asymptotic-frame calculation, not the unit-frequency
principal azimuthal screen \eqref{eq:kerr-azimuthal-screen}.  It gives the same
Schwarzschild term as \eqref{eq:schw-grazing-lambda} and displays how different
spin projections enter: \(J_p\) changes the eigenvalue split at first order,
whereas \(J_b\) and \(J_k\) rotate the instantaneous eigenvectors.

\paragraph{Periastron frame check.}
There is no conflict with the screen-eigenvalue invariance under local
\(SO(2)\) rotations.  For a fixed spacetime point and fixed null covector
\(\ell\), the DH screen matrix is a symmetric form on the Sachs quotient
\(\ell^\perp/\ell\).  If two screen representatives are related by
\begin{equation}
        s'_A=O_A{}^B s_B+\gamma_A\ell,\qquad O\in O(2),
        \label{eq:screen-quotient-transform}
\end{equation}
then Riemann antisymmetry gives
\begin{equation}
        R_{abcd}\ell^a s_A^{\prime b}\ell^c s_B^{\prime d}
        =O_A{}^C O_B{}^D R_{abcd}\ell^a s_C^b\ell^c s_D^d .
        \label{eq:screen-similarity-transform}
\end{equation}
The \(\gamma_A\ell\) pieces drop out because they put two copies of \(\ell\) in
an antisymmetric Riemann index pair.  Hence \(\tr\mathcal K\),
\(\det\mathcal K\), and the two eigenvalues are invariant under screen rotations
and null-gauge changes at fixed \(\ell\).

At closest approach in an equatorial prograde or retrograde geometry,
\(\bm J=J_p\hat{\bm p}\), the asymptotically transported Born frame and the Carter
principal frame are instead related, in the \((0,\hat{\bm k})\) plane, by a
longitudinal boost plus a screen rotation:
\begin{equation}
 \bar e_0=\cosh\eta\,e^{\rm P}_0+\sinh\eta\,e^{\rm P}_{\hat\phi},
 \qquad
 \bar e_{\hat k}=\sinh\eta\,e^{\rm P}_0+\cosh\eta\,e^{\rm P}_{\hat\phi},
 \qquad
 \eta=\frac{J_p}{Mb}+O\!\left(\frac{M}{b},\frac{a^2}{b^2}\right).
 \label{eq:principal-born-boost}
\end{equation}
Only the spin-odd longitudinal part of this weak-field principal-to-asymptotic
transformation is needed here.  The mass redshift and bending pieces multiply the
Schwarzschild screen by an extra \(O(M/b)\) and therefore enter only the discarded
\(O(M^2/b^4)\) remainder.  The sign of the displayed rapidity is fixed by
\(J_p=\pm Ma\).  Therefore the Born null vector satisfies
\begin{equation}
 \bar\ell^a:=\bar e_0^a+\bar e_{\hat k}^a
 =e^\eta\ell_{\rm P}^a+O\!\left(\frac{M}{b},\frac{a^2}{b^2}\right),
 \qquad \ell_{\rm P}^a=e_0^{{\rm P}a}+e_{\hat\phi}^{{\rm P}a} .
 \label{eq:born-null-rescaling}
\end{equation}
If \(\bar s_A=R_A{}^B s^{\rm P}_B+O(M/b,a^2/b^2)\) is the corresponding screen
rotation, then at \(l=0\)
\begin{equation}
 {\cal K}^{\rm Born}_{AB}(0)=
 e^{2\eta}R_A{}^CR_B{}^D{\cal K}^{\rm P}_{CD}
 +O\!\left(\frac{M^2}{b^4},\frac{MJ}{b^5},\frac{J^2}{b^6}\right).
 \label{eq:born-principal-screen-relation}
\end{equation}
The rotation leaves eigenvalues unchanged and the longitudinal null
normalization supplies the factor \(e^{2\eta}\).  Since
\(\lambda^{\rm P}_\pm=\pm(3M/b^3)[1+O(a^2/b^2)]\),
\begin{equation}
 \lambda^{\rm Born}_\pm(0)=e^{2\eta}\lambda^{\rm P}_\pm
 =\pm\left(\frac{3M}{b^3}+\frac{6J_p}{b^4}\right)
 +O\!\left(\frac{M^2}{b^4},\frac{MJ}{b^5},\frac{J^2}{b^6}\right),
 \label{eq:periastron-eigenvalue-match}
\end{equation}
which is exactly \eqref{eq:kerr-born-eigenvalues} at periastron, up to the
branch labels.  Thus eigenvalues are invariant under rotations of the same
screen at fixed null direction; the linear-in-spin Born term appears because the
asymptotic Born calculation uses a different null normalization/embedding from
the unit-frequency principal tetrad calculation.  Consequently the invariant
quantity in the scattering problem is the phase/delay integral constructed with
the chosen asymptotic normalization; the individual entries of
\({\cal K}_{\rm Born}(l)\) are pointwise representatives of that transport
problem, not local scalar observables by themselves.

The local Born eigenaxis angle is
\begin{equation}
        \chi_{\rm Born}(l)
        =\frac{J_b(3b^2-2l^2)+5J_kbl}
              {2Mb(b^2+l^2)}
        +O\!\left(\frac{J^2}{M^2b^2}\right).
        \label{eq:kerr-born-axis-angle}
\end{equation}
Consequently the signed asymptotic eigenaxis mismatch is
\begin{equation}
        \Theta^{\rm axis}_{-\infty\to+\infty}
        :=\lim_{L\to\infty}\bigl[\chi_{\rm Born}(L)-\chi_{\rm Born}(-L)\bigr]
        =0.
        \label{eq:kerr-born-axis-cancellation}
\end{equation}
Thus the magnetic-Weyl/local-axis effect does not by itself generate a new
nonzero asymptotic Faraday angle for an infinity-to-infinity grazing ray.  For
example, a polar grazing ray \((J_k\neq0, J_p=J_b=0)\) has a finite local axis
excursion \(|\chi_{\rm Born}|_{\max}=5|a|/(4b)\), but its signed endpoint
mismatch vanishes.

The surviving basis-covariant scattering quantity in this Born geometry is the
birefringent retardance.  To first order in spin the eigenvalues of
\eqref{eq:kerr-born-screen} are
\begin{equation}
        \lambda_\pm(l)=\pm\left[
        \frac{3Mb^2}{r(l)^5}+\frac{6J_p b}{r(l)^5}\right]
        +O\!\left(\frac{M^2}{r^4},\frac{MJ}{r^5},\frac{J_b^2+J_k^2}{M r^5}\right),
        \label{eq:kerr-born-eigenvalues}
\end{equation}
and \eqref{eq:signed-time-delay-general}, with \(\Delta t_{+-}:=t_- -t_+\), gives
\begin{align}
        \Delta t_{+-}^{\rm Kerr,Born}
        &=\zeta_{\rm DH}\int_{-\infty}^{\infty}
        \left[\frac{3Mb^2}{(b^2+l^2)^{5/2}}
        +\frac{6J_p b}{(b^2+l^2)^{5/2}}\right]\dd l
        \nonumber\\
        &=\frac{4\alpha}{45\pi}\lambda_e^2\frac{M}{b^2}
        \left[1+2\frac{\bm J\cdot(\hat{\bm k}\times\hat{\bm b})}{Mb}
        \right]
        +O\!\left(\zeta_{\rm DH}\frac{M^2}{b^3},
                  \zeta_{\rm DH}\frac{J^2}{Mb^4},
                  (\zeta_{\rm DH}R)^2\right).
        \label{eq:kerr-born-retardance}
\end{align}
The corresponding dimensionless optical phase retardance is
\begin{equation}
        \delta_{+-}^{\rm Kerr,Born}=\omega_\infty\Delta t_{+-}^{\rm Kerr,Born},
\end{equation}
within the same nondispersive low-frequency approximation.

\paragraph{Astrophysical scale.}
Restoring SI units with \(\lambda_e=3.8616\times10^{-13}\,\mathrm{m}\) and
\(GM_\odot/c^2=1.4766\,\mathrm{km}\), a photon of energy \(E\) has
\begin{align}
 \Delta t_{+-}^{\rm Kerr,Born}
 &\simeq 1.93\times10^{-43}\,\mathrm{s}
 \left(\frac{10M_\odot}{M}\right)
 \left(\frac{6M}{b}\right)^2
 \left[1+2\frac{a/M}{b/M}\right],\nonumber\\
 \delta_{+-}^{\rm Kerr,Born}
 &\simeq 1.76\times10^{-24}
 \left(\frac{E}{6\,\mathrm{keV}}\right)
 \left(\frac{10M_\odot}{M}\right)
 \left(\frac{6M}{b}\right)^2
 \left[1+2\frac{a/M}{b/M}\right].
 \label{eq:astrophysical-scale}
\end{align}
Thus a prograde ray with \(a/M=0.9\) and \(b=6M\) around a \(10M_\odot\) black
hole gives \(\delta_{+-}\simeq2.3\times10^{-24}\) at \(6\,\mathrm{keV}\), while
\(M=1.4M_\odot\), \(b=5M\), and \(a/M=0.3\) gives
\(\delta_{+-}\simeq2.0\times10^{-23}\).  IXPE-class gas-pixel X-ray polarimeters
operate in the few-keV band, and proposed future missions remain photon-limited
Stokes polarimeters rather than direct detectors of such microscopic phase
retardances \cite{Soffitta2021IXPEInstrument,Fabiani2018FutureXrayPolarimetry}.
Even compared with an illustrative \(10^{-2}\,\mathrm{rad}\) polarization-angle
scale, \eqref{eq:astrophysical-scale} is more than twenty orders of magnitude
smaller; its role here is scale-setting and code-benchmark normalization.

For an equatorial prograde or retrograde grazing ray, \(J_p=\pm Ma\), and the
branch delay acquires the spin-odd fractional correction \(\pm2a/b\).  This is
the asymptotic, frame-independent Kerr polarization signature supported by the
first-order DH screen reduction.  It is a low-frequency retardance of the
nondispersive surrogate; the signed eigenaxis rotation at infinity remains zero
at this order.

\paragraph{Scope of the cancellation.}
The cancellation used here is a far-zone statement about a two-ended weak-lensing
scattering problem with a specified asymptotic Sachs screen.  It follows from the
universal mass-spin metric in the interaction region and does not require an
exact Type-D argument.  The precise statement is the following.

\begin{proposition}[Far-zone eigenaxis cancellation and spin-odd retardance]
\label{prop:far-zone-axis-cancellation}
Consider any asymptotically flat stationary axisymmetric vacuum metric whose
weak-lensing region has the universal mass-spin far-zone form
\begin{equation}
 g_{00}=-1+\frac{2M}{r}+O(r^{-2}),\qquad
 g_{0i}=-\frac{2(\bm J\times\bm x)_i}{r^3}+O(r^{-3}),\qquad
 g_{ij}=\left(1+\frac{2M}{r}\right)\delta_{ij}+O(r^{-2}),
 \label{eq:universal-far-zone-metric}
\end{equation}
and let the unperturbed ray and the asymptotic screen be those of
\eqref{eq:kerr-born-ray}.  To first order in \(M/b\), \(J/(Mb)\), and
\(\zeta_{\rm DH}R\), the signed DH eigenaxis mismatch between the two
asymptotic screen frames vanishes,
\begin{equation}
        \Theta^{\rm axis}_{-\infty\to+\infty}=0,
        \label{eq:far-zone-axis-cancellation}
\end{equation}
whereas the branch-labelled birefringent retardance is
\begin{equation}
        \Delta t_{+-}=
        \frac{4\alpha}{45\pi}\lambda_e^2\frac{M}{b^2}
        \left[1+2\frac{\bm J\cdot(\hat{\bm k}\times\hat{\bm b})}{Mb}\right]
        +O\!\left(\zeta_{\rm DH}\frac{M^2}{b^3},
                  \zeta_{\rm DH}\frac{J^2}{Mb^4},
                  (\zeta_{\rm DH}R)^2\right).
        \label{eq:far-zone-retardance-theorem}
\end{equation}
Thus the leading conclusion depends only on the universal mass-spin far-zone
sector; higher multipoles and exact algebraic speciality enter beyond the
retained order.
\end{proposition}

\begin{proof}
The metric expansion \eqref{eq:universal-far-zone-metric} fixes the leading
Weyl electric and magnetic fields solely by \(M\) and \(\bm J\).  Projecting
\(R_{abcd}\ell^a s_A^b\ell^c s_B^d\) onto the parallel asymptotic screen
\((\hat{\bm b},\hat{\bm p})\) along the Born ray gives
\eqref{eq:kerr-born-screen}, independently of higher multipoles, because their
contributions enter at the displayed remainder order.  Write this screen matrix
as
\begin{equation}
        {\cal K}_{\rm Born}(l)=
        \begin{pmatrix}-A(l)&C(l)\\ C(l)&A(l)\end{pmatrix}+O(2),
\end{equation}
with
\begin{equation}
        A(l)=\frac{3Mb^2+6J_p b}{(b^2+l^2)^{5/2}},\qquad
        C(l)=-\frac{3b[J_b(3b^2-2l^2)+5J_kbl]}{(b^2+l^2)^{7/2}} .
\end{equation}
The local diagonalizing angle in this transported screen satisfies
\(\chi(l)=-C(l)/(2A_0(l))+O(J^2/(M^2b^2))\), where
\(A_0(l)=3Mb^2/(b^2+l^2)^{5/2}\).  Hence
\begin{equation}
        \chi(l)=
        \frac{J_b(3b^2-2l^2)+5J_kbl}{2Mb(b^2+l^2)}
        +O\!\left(\frac{J^2}{M^2b^2}\right),
\end{equation}
so the two endpoint limits are equal:
\(\lim_{l\to+\infty}\chi(l)=\lim_{l\to-\infty}\chi(l)=-J_b/(Mb)\).  In the
parallel asymptotic screen \(\varpi_l=0\), and in any other screen the extra
connection term in \eqref{eq:transport-covariant-eigenaxis-mismatch} changes by
exactly the compensating endpoint frame rotation.  Therefore the invariant
signed endpoint mismatch is zero.  This is the geometric protection mechanism:
the magnetic Weyl field can tilt the instantaneous local eigenbasis, but in a
two-ended asymptotically flat scattering problem the endpoint curvature
vanishes, the asymptotic eigenline is degenerate, and the transported endpoint
comparison removes the frame-dependent tilt.

The eigenvalues are
\(\lambda_\pm=\pm\sqrt{A(l)^2+C(l)^2}\).  The off-diagonal tilt term \(C\) enters
only quadratically in the eigenvalue split, while the \(J_p\) part of \(A\) is
linear.  Thus, to the retained order,
\begin{equation}
        \lambda_+(l)-\lambda_-(l)=
        2\left[\frac{3Mb^2}{(b^2+l^2)^{5/2}}
        +\frac{6J_p b}{(b^2+l^2)^{5/2}}\right].
\end{equation}
Substitution into \eqref{eq:signed-time-delay-general} and the elementary integral
\(\int_{-\infty}^{\infty}(b^2+l^2)^{-5/2}\,\dd l=4/(3b^4)\) gives
\eqref{eq:far-zone-retardance-theorem}.
\end{proof}

The present Kerr statement should therefore be read as follows.  Daniels--Shore
established the Kerr DH propagation setting in a stationary-frame calculation.
Here the unit-frequency principal azimuthal screen has a spin-even split and a
spin-odd local eigenbasis tilt, while the asymptotically normalized Born screen
has the invariant periastron relation \eqref{eq:periastron-eigenvalue-match} and
the spin-odd integrated retardance \eqref{eq:kerr-born-retardance}.  A full Kerr
polarization calculation must combine \eqref{eq:branch-Hamiltonians} with the
transport-covariant update \eqref{eq:kerr-transport-rate} along the chosen ray.

\section{Interpretation and dynamical scope}
\label{sec:interpretation}

For a supplied local nondispersive constitutive tensor inheriting the reflection
symmetry, \(P_{\rm par}\) is the exact Tamm--Rubilar polynomial of that supplied
tensor.  It is a physical principal polynomial only for a genuine nondispersive
system.  For DH input it is the root polynomial of a low-frequency surrogate;
nonlinear powers generated by the cubic Tamm--Rubilar contraction are
finite-amplitude conditioning tests, not higher-loop QED corrections.

The Schwarzschild calculation is therefore an exact algebraic factorization of
the supplied DH surrogate and a first-order physical calibration against the
known DH cone shift.  The Ricci-flat screen reduction supplies the propagation
layer: branch Hamiltonians, local margin bounds, the Schwarzschild grazing-ray
split, and the Kerr Born branch-delay benchmark.  Generic full-parity rational tests and the flat-slice shift-potential geometry are retained as supplementary reproducibility and ADM/Codazzi/Gauss consistency checks.

A production code should call the full invariant Tamm--Rubilar contraction
whenever the supplied tensor is not in the SSSW-frame subclass, attach the DH EFT
status in \cref{subsec:qed-dispersion-error} to every curvature-based root
evaluation, and pass only accepted, labelled branches to Hamiltonian transport.

\section{Conclusion}
\label{sec:conclusion}

The paper gives a local algebraic diagnostic for supplied nondispersive
constitutive tensors and calibrates it on the DH curvature coupling.  The parity
theorem proves the exact block form \(G_{++}\oplus G_{--}\) and the even-in-\(z\)
structure of the full Tamm--Rubilar polynomial; the vanishing of the cubic
coefficient belongs only to the reconstructed SSSW-frame subclass.  The
six-variable meridional polynomial is retained as an auditable benchmark, not as
a complete parity-reduced theory.

For Schwarzschild DH input the full local tensor factorizes and reproduces the
standard low-frequency result: radial propagation is unshifted and tangential
propagation splits into two first-order polarization branches.  The Ricci-flat
screen reduction then turns this local splitting into branch Hamiltonians and a
ray-integrated weak-lensing delay split.

The Kerr weak-lensing calculation is the main rotating-spacetime benchmark.  It
separates two effects that should not be conflated.  The local magnetic-Weyl
angle \(\chi_K=-3a\cos\theta/(2r)+O(a^3/r^3)\) tilts the instantaneous principal
screen eigenbasis but does not, in the stated infinity-to-infinity Born setup,
produce a nonzero gauge-invariant endpoint Faraday angle at leading order.  The
nonzero frame-independent low-frequency observable is instead the branch-delay
split, or phase retardance \(\delta_{+-}=\omega_\infty\Delta t_{+-}\), with the
spin-odd correction proportional to
\(\bm J\cdot(\hat{\bm k}\times\hat{\bm b})/b^3\).

The result is a reproducible local layer between a supplied constitutive tensor
and later global ray-tracing or closure calculations: construct \(P_x(q)\),
certify real ADM-separated branches and margins, verify the DH EFT domain when
applicable, and only then evolve the accepted Hamiltonian branches.

\appendix

\section{Fixed-support meridional connection check}
\label{app:fixed-support}

For completeness we record the kinematic check behind the meridional support
used in the main text.  With
\(E_1=[01]\), \(E_2=[02]\), \(E_3=[12]\), the ADM-normal meridional connection
matrices, after removing the inessential local spatial \(SO(2)\) rotation, are
\[
 \Gamma_{\hat r}=\begin{pmatrix}0&0&-\mathcal B\\0&0&\mathcal A\\-\mathcal B&\mathcal A&0\end{pmatrix},
 \qquad
 \Gamma_{\hat\theta}=\begin{pmatrix}0&0&-\mathcal D\\0&0&\mathcal B\\-\mathcal D&\mathcal B&0\end{pmatrix}.
\]
A common fixed rank-one meridional ansatz must be a common invariant line of
these two real symmetric matrices.  Writing
\(M(\alpha,\beta)=\begin{psmallmatrix}0_{2\times2}&\nu\\ \nu^T&0\end{psmallmatrix}\)
with \(\nu=(-\beta,\alpha)^T\), a common line requires
\((-\mathcal B,\mathcal A)^T\parallel(-\mathcal D,\mathcal B)^T\), hence
\[
        \mathcal A\mathcal D-\mathcal B^2=0 .
\]
Away from this nongeneric locus, a fixed ansatz containing \([01]\) is forced to
include \([12]\) and then \([02]\), giving the full meridional support
\(\spanop\{[01],[02],[12]\}\).  This is only a support-minimality check; it is
not a dynamical area-metric field equation.

\section{Flat-slice shift-potential test geometry}
\label{app:shift-background}
\label{sec:background}

The flat-slice axisymmetric shift-potential geometry serves strictly as a
formal tensorial consistency testbed, not as an astrophysical solution.  It gives
a differentiable reflection-symmetric background with nonzero meridional shear,
\(K_{ij}=D_iD_j\Phi\), on which the ADM momentum identity, Codazzi relations,
and Gauss equation test the parity-block projection dictionary.  The metric data
are written in
geometric units $G=c=1$ with ADM sign convention
\begin{equation}
  \alpha=1,\qquad \gamma_{ij}=\delta_{ij},\qquad \beta^\flat=-\dd\Phi .
\end{equation}
The spatial slices are flat and time independent, so
\begin{equation}
  K_{ij}=-D_{(i}\beta_{j)}=D_iD_j\Phi .
\end{equation}
The numerical examples use a standard smooth radial top-hat transition
function
\begin{equation}
 S_{\rho,\sigma}(r)
 =\frac{\tanh[\sigma(r+\rho)]-\tanh[\sigma(r-\rho)]}{2\tanh(\sigma\rho)} .
 \label{eq:smooth-top-hat}
\end{equation}
The polar axis is the symmetry axis and
\begin{equation}
  \Phi(r,\theta)=v\,r\,g(r)\cos\theta ,
\end{equation}
where
\begin{equation}
  f(r)=1-S_{\rho,\sigma}(r),
  \qquad
  g(r)=\frac{1}{r}\int_0^r f(s)\,\dd s,
  \qquad
  g'(r)=\frac{f(r)-g(r)}{r} .
\end{equation}
The \(C^\infty\) tanh top-hat \eqref{eq:smooth-top-hat} is chosen for three
practical reasons: the parameters \(\rho\) and \(\sigma\) independently control
the transition radius and width, the profile and its derivatives are analytic,
and the resulting potential generates a localized transition region of nonzero mixed
meridional shear.  In the present calculation the profile is used only to define
the scalar potential \(\Phi=v r g(r)\cos\theta\) and the associated extrinsic-curvature
data \(K_{ij}=D_iD_j\Phi\); no trajectory or source model is specified.

The relevant orthonormal components of the extrinsic curvature are
\begin{equation}
 \mathcal A:=K_{\hat r\hat r}=v f'(r)\cos\theta,
 \qquad
 \mathcal B:=K_{\hat r\hat\theta}=-v g'(r)\sin\theta,
 \qquad
 \mathcal D:=K_{\hat\theta\hat\theta}=K_{\hat\phi\hat\phi}=v g'(r)\cos\theta .
\end{equation}
The equality $\mathcal D=vg'(r)\cos\theta$ follows from $(rg)'=f$ and
the standard spherical-coordinate second derivatives of $\Phi$.  The meridional determinant
\begin{equation}
 \mathcal A\mathcal D-\mathcal B^2
\end{equation}
is the same ADM-normal meridional $SO(2)$-invariant quantity that appeared in the fixed rank-one closure condition.  At the
equator it equals $-v^2[g'(r)]^2$ wherever $vg'(r)\ne0$, so the chosen
background genuinely tests the full meridional ansatz space rather than a
rank-one subcase.

The feature used in the Drummond--Hathrell projection is the commuting-derivative identity.  Because $K_{ij}=D_iD_j\Phi$ on a flat slice, the ADM
momentum density vanishes identically:
\begin{equation}
  8\pi j_i=D_j(K^j{}_i-\delta^j{}_iK)=D_i(D^2\Phi)-D_i(D^2\Phi)=0 .
\end{equation}
This identity, together with the Codazzi equations, is the tensorial selection
rule checked in the supplementary flat-slice implementation.

\section{Compact symbolic derivation of the restricted quartic}
\label{app:quartic-derivation}

This appendix gives the compact construction whose expanded result is checked in
\cref{app:expanded-coefficients} and supplied in the symbolic supplement.  Let
the ordered spacetime indices be \(a,b=0,1,2,3\), let spatial indices be
\(i,j=1,2,3\), and use \(\epsilon^{123}=+1\).  In the SSSW parametrization of
\cref{sec:meridional-sector} the reconstructed area-metric constitutive tensor is represented by
\begin{align}
G^{0i0j}&=-U^{ij},\label{eq:app-G-0011}\\
G^{0ijk}&=\Delta(\delta^i{}_m+W^i{}_m)\epsilon^{mjk},\label{eq:app-G-0ijk}\\
G^{ijkl}&=\Delta^2\epsilon^{ijm}\epsilon^{kln}L_{mn},\label{eq:app-G-ijkl}
\end{align}
with \(U,L,W,\Delta\) defined in \eqref{eq:sssw-definitions}--\eqref{eq:mixed-block-W}.  The frame condition \(WU=(WU)^T\) is imposed by the definitions of \(a,b,c,d\); it is not an additional field equation.

The principal tensor is the Tamm--Rubilar/SSSW cubic contraction of the area metric,
\begin{equation}
\mathcal G^{abcd}[G]
=-\frac{1}{24}\,\epsilon_{mnpq}\epsilon_{rstu}\,
G^{mn r(a}G^{b|ps|c}G^{d)qtu},
\label{eq:app-TR-tensor}
\end{equation}
where \(\epsilon_{mnpq}\) is the four-dimensional Levi-Civita tensor density, parentheses denote symmetrization over the free indices \(a,b,c,d\), and vertical bars exclude the enclosed indices from that symmetrization.  This notation is kept distinct from the frequency variable \(\omega\) in the covector below.  The raw characteristic polynomial is
\begin{equation}
\widetilde P_\Psi(q):=\mathcal G^{abcd}[G(\Psi)]q_aq_bq_cq_d,
\qquad q_a=(\omega,x,y,z).
\label{eq:app-raw-P}
\end{equation}
Substituting \eqref{eq:app-G-0011}--\eqref{eq:app-G-ijkl} into \eqref{eq:app-TR-tensor} gives
\begin{equation}
\widetilde P_\Psi
=\Delta^2\omega^4+\widetilde C_2\omega^2+\widetilde C_1\omega+\widetilde C_0,
\label{eq:app-raw-quartic}
\end{equation}
with no cubic term.  The cancellation of the \(\omega^3\) coefficient is a useful invariant check: it follows from the SSSW frame condition \(WU=(WU)^T\), equivalently from the absence of the corresponding antisymmetric magneto-electric trace in the reconstructed weakly birefringent sector.

On the symmetric-positive branch \(\Delta^2=uw-s^2>0\), multiplication of the characteristic polynomial by a positive nonzero scalar does not change the characteristic set.  The monic polynomial of \cref{thm:restricted-polynomial} is therefore
\begin{equation}
P_\Psi(q):=\frac{\widetilde P_\Psi(q)}{\Delta^2}
=\omega^4+C_2\omega^2+C_1\omega+C_0,
\qquad
C_i:=\frac{\widetilde C_i}{\Delta^2}.
\label{eq:app-monic-P}
\end{equation}
The compact definition \eqref{eq:app-TR-tensor}--\eqref{eq:app-monic-P} is the
conceptual derivation.  The fully expanded coefficients are supplied as a
machine-readable supplement, while \cref{app:expanded-coefficients} records the
short specializations that fix their relative signs.

\section{Coefficient checks for the restricted quartic}
\label{app:expanded-coefficients}

The expanded six-variable coefficients \(C_2,C_1,C_0\) are not printed in the
article.  They are generated from the invariant contraction
\eqref{eq:app-TR-tensor}--\eqref{eq:app-monic-P} and supplied in the symbolic
supplement, where they can be compared directly against independent
implementations.  The article retains the compact checks that are most useful
for auditing signs and normalizations.

In the metric vacuum,
\begin{equation}
 u=w=\ell=1,\qquad s=a=b=c=d=0,
\end{equation}
one obtains
\begin{equation}
 C_2=-2(x^2+y^2+z^2),\qquad C_1=0,
 \qquad C_0=(x^2+y^2+z^2)^2,
\end{equation}
so that \(P_\Psi=(\omega^2-x^2-y^2-z^2)^2\).  If the mixed sector is set to zero,
\(a=b=c=d=0\), then
\begin{equation}
 C_2^{\rm sym}=-\ell Q-\Delta^2(x^2+y^2)-(u+w)z^2,
\end{equation}
\begin{equation}
 C_0^{\rm sym}=\Delta^2\{\ell S_\perp Q+z^2(\ell S_\perp+Q)+z^4\},
 \qquad S_\perp=x^2+y^2,
 \qquad Q=ux^2+2sxy+wy^2 .
\end{equation}
These identities follow from \(G^{0i0j}=-U^{ij}\),
\(G^{0ijk}=\Delta(\delta^i{}_m+W^i{}_m)\epsilon^{mjk}\),
\(G^{ijkl}=\Delta^2\epsilon^{ijm}\epsilon^{kln}L_{mn}\), and
\(\epsilon^{123}=+1\).  They fix the relative signs of the \(\omega^2\) and
purely spatial terms.  The supplement additionally checks the scalar-envelope
specialization and the SSSW-frame cancellation of the cubic coefficient.

\section{Projection dictionary and full Drummond--Hathrell parity blocks}

\label{app:dictionary}

This appendix collects the sign and normalization material behind the main DH
application.  The convention is
\(H^{ab}=(1/2)\chi^{abcd}F_{cd}\) with metric signature \(-+++\).  If
\(\chi^{abcd}_{\partial}=\partial H^{ab}/\partial F_{cd}\) is formed by
antisymmetric differentiation, then \(\chi^{abcd}=2\chi^{abcd}_{\partial}\).  In
the ordered independent bivector basis
\(([01],[02],[03],[23],[31],[12])\), the Maxwell vacuum is
\(\diag(-1,-1,-1,+1,+1,+1)\).  This is the origin of the electric-electric minus
signs below.

For the restricted meridional basis \(E_1=[01]\), \(E_2=[02]\), \(E_3=[12]\),
the first-order near-metric dictionary is
\begin{equation}
\begin{aligned}
 \phi_1&=-\delta\chi^{0101},&
 \phi_2&=-\sqrt2\,\delta\chi^{0102},&
 \phi_4&=-\delta\chi^{0202},\\
 \phi_{12}&=+\delta\chi^{1212},&
 \phi_{17}&=+\sqrt2\,\delta\chi^{0112},&
 \phi_{16}&=+\sqrt2\,\delta\chi^{0212}.
\end{aligned}
\label{eq:sssw-projection}
\end{equation}
It projects a supplied four-index tensor onto the selected meridional SSSW-frame
variables; it is not a nonlinear inverse from an arbitrary finite constitutive
tensor to SSSW variables.

Substituting the local DH tensor \eqref{eq:dh-chi} gives
\begin{align}
\phi^{\rm DH}_1
&=\frac{2}{m_e^2}\left[-2c_RR+c_{\rm Ric}(R^{00}-R^{11})+4c_{\rm Riem}R^{0101}\right],\nonumber\\
\phi^{\rm DH}_2
&=\frac{2\sqrt2}{m_e^2}\left[-c_{\rm Ric}R^{12}+4c_{\rm Riem}R^{0102}\right],\nonumber\\
\phi^{\rm DH}_4
&=\frac{2}{m_e^2}\left[-2c_RR+c_{\rm Ric}(R^{00}-R^{22})+4c_{\rm Riem}R^{0202}\right],\nonumber\\
\phi^{\rm DH}_{12}
&=-\frac{2}{m_e^2}\left[2c_RR+c_{\rm Ric}(R^{11}+R^{22})+4c_{\rm Riem}R^{1212}\right],
\label{eq:dh-symmetric}
\end{align}
and the selected mixed entries are
\begin{align}
\phi^{\rm DH}_{17}
&=\sqrt2\,\delta\chi^{0112}_{\rm DH}
 =\frac{2\sqrt2}{m_e^2}\left[c_{\rm Ric}R^{02}-4c_{\rm Riem}R^{0112}\right],\nonumber\\
\phi^{\rm DH}_{16}
&=\sqrt2\,\delta\chi^{0212}_{\rm DH}
 =-\frac{2\sqrt2}{m_e^2}\left[c_{\rm Ric}R^{01}+4c_{\rm Riem}R^{0212}\right].
\label{eq:dh-mixed}
\end{align}

\subsection{Transverse parity-odd block of the Drummond--Hathrell projection}
\label{subsec:transverse-dh}

The complementary parity block uses \(V_- =\spanop\{[03],[13],[23]\}\).  Its raw
DH entries are
\begin{align}
\delta\chi^{0303}_{\rm DH}
&=\frac{2}{m_e^2}\left[2c_RR-c_{\rm Ric}(R^{00}-R^{33})-4c_{\rm Riem}R^{0303}\right],
\label{eq:dh-trans-0303}\\
\delta\chi^{1313}_{\rm DH}
&=-\frac{2}{m_e^2}\left[2c_RR+c_{\rm Ric}(R^{11}+R^{33})+4c_{\rm Riem}R^{1313}\right],
\label{eq:dh-trans-1313}\\
\delta\chi^{2323}_{\rm DH}
&=-\frac{2}{m_e^2}\left[2c_RR+c_{\rm Ric}(R^{22}+R^{33})+4c_{\rm Riem}R^{2323}\right],
\label{eq:dh-trans-2323}
\end{align}
and
\begin{align}
\delta\chi^{0313}_{\rm DH}
&=-\frac{2}{m_e^2}\left[c_{\rm Ric}R^{01}+4c_{\rm Riem}R^{0313}\right],
\label{eq:dh-trans-0313}\\
\delta\chi^{0323}_{\rm DH}
&=-\frac{2}{m_e^2}\left[c_{\rm Ric}R^{02}+4c_{\rm Riem}R^{0323}\right],
\label{eq:dh-trans-0323}\\
\delta\chi^{1323}_{\rm DH}
&=-\frac{2}{m_e^2}\left[c_{\rm Ric}R^{12}+4c_{\rm Riem}R^{1323}\right].
\label{eq:dh-trans-1323}
\end{align}
With the same electric-magnetic sign convention as \eqref{eq:sssw-projection},
\begin{equation}
\phi_3=-\delta\chi^{0303},\qquad
\phi_5=+\delta\chi^{1313},\qquad
\phi_6=+\delta\chi^{2323},
\label{eq:transverse-diag-vars}
\end{equation}
and
\begin{equation}
\phi_{13}=\sqrt2\,\delta\chi^{0313},\qquad
\phi_{14}=\sqrt2\,\delta\chi^{0323},\qquad
\phi_{15}=\sqrt2\,\delta\chi^{1323}.
\label{eq:transverse-mixed-vars}
\end{equation}
These formulae fix only a coordinate convention for comparing the two parity
blocks; the full calculation uses the raw bivector tensor.

For the formal flat-slice shift-potential background of \cref{app:shift-background},
\(K_{ij}=D_iD_j\Phi\) on a flat slice.  The ADM momentum identity gives
\(R_{0i}=0\), and the Codazzi relation gives
\begin{equation}
R_{\hat0\hat i\hat j\hat k}=D_{\hat j}K_{\hat k\hat i}-D_{\hat k}K_{\hat j\hat i}
=[D_{\hat j},D_{\hat k}]D_{\hat i}\Phi=0 .
\label{eq:transverse-codazzi-zero}
\end{equation}
Consequently
\begin{equation}
 \delta\chi^{0313}_{\rm DH}=\delta\chi^{0323}_{\rm DH}=0,
 \label{eq:transverse-me-zero}
\end{equation}
while the purely spatial off-diagonal entry can remain nonzero.  With
\(K_{13}=K_{23}=0\), the flat-slice Gauss relation gives
\begin{equation}
R^{1323}=K_{12}K_{33}=\mathcal B\mathcal D,
\label{eq:transverse-spatial-mix}
\end{equation}
so the transverse parity block is generically populated even when the selected
transverse magneto-electric slots vanish.  The corresponding full DH supplied
state may be recorded as
\begin{equation}
\Psi_{\rm DH}^{\rm full}=
\left(\phi^{\rm DH}_1,\phi^{\rm DH}_2,\phi^{\rm DH}_4,
\phi^{\rm DH}_{12},0,0\;\middle|\;
\phi^{\rm DH}_3,\phi^{\rm DH}_5,\phi^{\rm DH}_6,0,0,\phi^{\rm DH}_{15}\right),
\label{eq:full-DH-state}
\end{equation}
where \(\phi^{\rm DH}_{15}\) includes the Ricci and Gauss terms in
\eqref{eq:dh-trans-1323} and \eqref{eq:transverse-spatial-mix}.  This appendix is
a sign and selection-rule check only; it is not a gravitational source model and
not the full dispersive Drummond--Hathrell photon-propagation problem.

\section{Supplementary symbolic framework}
\label{app:verification-bundle}

The supplementary archive \texttt{cqg\_symbolic\_supplement\_v1.zip} contains a
compact symbolic framework supporting the algebraic claims made in the text.  The
archive is intended as a reproducibility supplement; the invariant formulae in
the article remain the mathematical definitions.

The three verification files are as follows.
\begin{description}
\item[\texttt{verify\_restricted\_sssw\_tr.py}] Reconstructs the restricted SSSW
area metric from \((u,s,w,\ell,a,b,c,d)\), forms the Tamm--Rubilar tensor
\eqref{eq:app-TR-tensor}, contracts it with \(q_a=(\omega,x,y,z)\),
monic-normalizes the result, and writes the expanded coefficients
\(C_2,C_1,C_0\) to \texttt{outputs/sssw\_C2\_C1\_C0\_expanded.txt}.  The same
script checks the metric-vacuum specialization, the mixed-off coefficient checks
of \cref{app:expanded-coefficients}, the first-order near-metric expansion, and
the SSSW-frame cancellation of the cubic coefficient.

\item[\texttt{verify\_full\_parity\_generic.py}] Implements the full local
parity-invariant generator of \cref{thm:complete-parity-polynomial} in the
basis \(([01],[02],[12] \mid [03],[13],[23])\).  It checks the metric vacuum,
the Schwarzschild Drummond--Hathrell factorization, the exact even-in-\(z\)
structure, and a rational near-metric supplied tensor with both \(G_{++}\) and
\(G_{--}\) populated for which \(C_3\ne0\).

\item[\texttt{verify\_flat\_slice\_shift.py}] Evaluates the flat-slice
shift-potential consistency test.  It checks the ADM momentum identity,
Codazzi selection rule, meridional support determinant, and Gauss-equation
selection rule used in the projection dictionary.
\end{description}

The archive also contains the shared utility file \texttt{tr\_utils.py}, the
runner \texttt{run\_all.py}, \texttt{requirements.txt}, a \texttt{README.md},
and the generated reports in the \texttt{outputs/} directory.  The complete test
suite is reproduced by running
\begin{verbatim}
python -m pip install -r requirements.txt
python run_all.py
\end{verbatim}
from the unpacked archive directory.  Each script raises an exception if a check
fails, and successful runs regenerate the report files in \texttt{outputs/}.

\paragraph{Code availability.}
The symbolic verification bundle supporting the algebraic checks in
\cref{app:verification-bundle} is supplied with the arXiv version as the
ancillary file \texttt{cqg\_symbolic\_supplement\_v1.zip}.  The archive
contains the restricted SSSW Tamm--Rubilar reconstruction, the full
parity-invariant generator, generic rational near-metric tests with
\(C_3\ne0\), the Schwarzschild factorization check, and the flat-slice
ADM/Codazzi/Gauss consistency tests.  The archive SHA256 checksum is
\texttt{d11b90be47a3ae0adf335f612940a330a8bd6727f8e65b293c530ac8f8dca140}.

\printbibliography

\end{document}